\documentclass[lettersize,journal]{IEEEtran} 
\usepackage{amsfonts,amsthm,amsmath,amssymb,mathtools}
    
    \DeclareMathOperator{\cov}{cov}
    \DeclareMathOperator{\E}{E}
    \DeclareMathOperator{\ent}{H}
    \DeclareMathOperator{\FE}{FE} 
    \DeclareMathOperator{\loglik}{LL} 
    \DeclareMathOperator{\LL}{LL}
    \DeclareMathOperator{\med}{median} 
    \DeclareMathOperator{\pa}{pa}
    \DeclareMathOperator{\Scope}{scope}
    \DeclareMathOperator{\tr}{tr}
    \DeclareMathOperator{\var}{var}
    \DeclareMathOperator{\mtov}{vec}

\usepackage{algorithm}
\usepackage[noend]{algpseudocode} 
\usepackage{array}
\usepackage{bm}
\usepackage[caption=false,font=normalsize,labelfont=sf,textfont=sf]{subfig}
\usepackage{textcomp}
\usepackage{stfloats}
\usepackage{url}
\usepackage{verbatim}
\usepackage{graphicx}
\graphicspath{ {figures} }
\usepackage[]{hyperref}
\usepackage{cite}
\usepackage{subfig}

\theoremstyle{definition}
\newtheorem{definition}{\textit{Definition}}[section]

\newtheorem{theorem}{Theorem}

\usepackage[normalem]{ulem} 
\usepackage{xcolor}
\definecolor{changecolor}{RGB}{192,64,0}

\usepackage{refcount}

\begin{document}

\title{Adapting cluster graphs for inference of continuous trait evolution on phylogenetic
networks}

\author{
Benjamin Teo and Cécile Ané
\thanks{B. \!Teo is with the Department of Statistics, University of Wisconsin-Madison,
Madison, WI, USA}
\thanks{C. \!Ané is with the Departments of Statistics and Botany, University of
Wisconsin-Madison, Madison, WI, USA.}
}

\markboth{}%
{Shell \MakeLowercase{\textit{et al.}}}

\IEEEpubid{}

\maketitle

\begin{abstract}
Dynamic programming approaches have long been applied to fit models of
univariate and multivariate trait evolution on phylogenetic trees
for discrete and continuous traits, and more recently adapted to phylogenetic
networks with reticulation.
We previously showed that various trait evolution models on a network can be readily cast as
probabilistic graphical models, so that likelihood-based
estimation can proceed
efficiently via belief propagation on an associated clique tree.
Even so, exact likelihood inference can grow computationally prohibitive for large complex networks.
Loopy belief propagation can similarly be applied to these settings, using non-tree cluster
graphs to optimize a factored energy approximation to the log-likelihood, and may provide a
more practical trade-off between estimation accuracy and runtime.
However, the influence of cluster graph structure on this trade-off is not precisely
understood.
We conduct a simulation study using the Julia package \texttt{PhyloGaussianBeliefProp} to
investigate how varying maximum cluster size affects this trade-off for Gaussian trait
evolution models on networks.
We discuss recommended choices for maximum cluster size, and prove the equivalence of
likelihood-based and factored-energy-based parameter estimates for the homogeneous Brownian
motion model.
\end{abstract}

\begin{IEEEkeywords}
continuous trait, linear Gaussian, admixture graph, belief propagation, cluster graph,
approximate inference
\end{IEEEkeywords}

\section{Motivation}
\IEEEPARstart{L}{oopy} belief propagation (LBP) provides a means of parameter estimation by
approximate inference \cite[Sec.~20.5.1]{koller2009probabilistic}.
LBP extends belief propagation (BP) in an appealing way that retains BP's simple formulation
and intuitive message-passing rules \cite[Sec.~11.3.2]{koller2009probabilistic}.
Although LBP faces non-convergence issues and approximation error,
especially when the ``simplest'' cluster graphs are used \cite{malioutov2006walk,kschischang2001:factor},
its performance can seemingly be improved by constructing alternative cluster graphs
\cite{kamper2019:GaBP-m,teo2024leveraging}.
However, such alternatives have received less attention, in part because of limited
automatic construction procedures \cite{streicher2017graph} and because the trade-offs
between different cluster graphs are difficult to characterize precisely.

We take up the challenge of systematically investigating how cluster graph structure affects
the performance of LBP for fitting Gaussian trait models on phylogenetic networks
\cite{mitov2020fast, teo2024leveraging}.
For the scope of this study, we focus on maximum likelihood estimation for the Brownian
motion model \cite[Ch.~3,5]{harmon2019} on three networks of varying topological
complexity, taken from the admixture graph and ancestral recombination
graph literature \cite{2020Lipson,lewanski2024:arg}.
We use the join-graph structuring algorithm to generate cluster
graphs of varying maximum cluster size $k$ \cite{mateescu2010join}, and address the problem of choosing $k$ to strike
a good balance between computational cost and estimation accuracy.
We run LBP with a fixed strategy for scheduling messages and initializing cluster graph beliefs
(Sec.~\ref{ch4:sched&reg}), though other choices could be explored to improve convergence.

Our objective is to assess the potential of LBP as a tool for more scalable, approximate
inference of Gaussian trait models on evolutionary networks,
which can be valuable when
parameter estimation requires iterative numerical optimization
\cite{Pybus2012,2023Bartoszek-modelselection-mvOU}.
We hope that our findings initiate or encourage the design of new cluster graphs, message
schedules, and belief initializations that better suit the phylogenetic context.

\section{Background}
\noindent
This section provides a brief overview of the key concepts involved for using LBP to fit
Gaussian trait models on phylogenetic networks.
For a more detailed exposition, see \cite{teo2024leveraging}.

\subsection{Trait evolution on phylogenetic networks}\label{ch4:traitevo}
\noindent
A phylogeny is a graph that describes the shared genealogy for a set of entities, typically
but not necessarily biological (e.g., species \cite{rose2021} or languages \cite{sagart2019dated}).
Nodes represent these entities and their common ancestors, while edges represent lineages
(Fig.~\ref{fig:net}).
Rooting the phylogeny by specifying an oldest node fixes the direction of
time and evolution along each edge.
As time cannot go backwards, a rooted phylogeny cannot have any directed cycle.

\begin{definition}[Rooted phylogenetic network {\cite[Ex.~2]{teo2024leveraging}}]
    A \emph{rooted phylogenetic network} is a connected, directed acyclic graph (DAG)
    $N=(V,E)$ with a single \emph{root} (a node with no parents) and taxon-labeled leaves.
    Nodes with at most one parent are called \emph{tree nodes}, and their in-edges are called
    \emph{tree edges}.
    Nodes with multiple parents represent populations with mixed ancestry. They are called
    \emph{hybrid nodes}, and their in-edges are called \emph{hybrid edges}.
    Each hybrid edge $e=(u,h)$ is assigned an \emph{inheritance weight} $\gamma(e)>0$ that
    represents the proportion of the genome in hybrid $h$ that was inherited from its
    parent $u$.
    The inheritance weights for each hybrid node must sum to 1.
\end{definition}

\newpage
Trait evolution is widely modeled over a rooted phylogeny as a forwards-in-time stochastic
process that progresses along the graph edges, culminating in the observed traits at
the leaves (Fig.~\ref{fig:traitevol}).
Evolutionary changes along separate edges are typically assumed
independent given their start states \cite{steel2016:phylogeny}.
If not, such as to model interacting populations or the variability of gene
histories
\cite{2017ManceauLambertMorlon-interactinglineages,2021Rabier-snappnet},
trait evolution can be modeled over a
supergraph of the phylogeny \cite{teo2024leveraging}.
Either way, trait evolution models can be described by a graphical model with
conditional distributions local to each edge, which are induced by the stochastic process
\cite{Mitov2019,2021Bastide-HMC,bastide2023cauchy,teo2024leveraging}.

\subsection{Gaussian trait models}\label{ch4:gaussianmodels}
\noindent
Gaussian trait models remain a workhorse in the study of continuous trait evolution
\cite{HoAne2014_phylolm}, 
and are often characterized by linear Gaussian models along each edge \cite{mitov2020fast}.

\begin{definition}[Linear Gaussian model {\cite[Sec.~3(a)]{teo2024leveraging}}]
    \label{ch4:def:lineargaussian}
    A \emph{linear Gaussian} model on a directed tree edge $(u,v)$ 
    assumes that the $p$-dimensional trait
    $X_v$ at node $v$ 
    has the following conditional distribution, given the trait $X_u$
    of its single parent $u$: 
    \begin{equation}\label{eq:lineargaussian-1edge}
        X_v\mid X_u\sim\mathcal{N}(\bm{q}_v X_u+\omega_v, \bm{V}_v)
    \end{equation}
    where the $p\times p$ actualization matrix $\bm{q}_v$, length-$p$ trend vector $\omega_v$,
    and $p\times p$ covariance matrix $\bm{V}_v$ do not depend on $X_u$.

    \noindent
    If $v$ has multiple parents $\pa(v)=\{u_1,\dots,u_m\}$, whose stacked traits
    $\mtov(\begin{bmatrix}X_{u_1}\dots X_{u_m}\end{bmatrix})$
    we denote by $X_{\pa(v)}$,
    then the \emph{node family} $\{v\}\cup\pa(v)$
    has a linear Gaussian model if
    \begin{equation}\label{eq:lineargaussian-1hybrid}
        X_v\mid X_{\pa(v)}\sim\mathcal{N}(\bm{q}_v X_{\pa(v)}+\omega_v, \bm{V}_v)
    \end{equation}
    with $\bm{q}_v$, $\omega_v$, and $\bm{V}_v$ independent of $X_{\pa(v)}$, where
    $\bm{q}_v$ is now of size $p\times (mp)$.
    A DAG for which every node family has a linear Gaussian model is known as a
    \emph{linear Gaussian network}.
\end{definition}

If $v$ has multiple parent edges $e_k=(u_k,v)$,
it is reasonable to assume that $X_v$ is a function of
$(X_{\underline{e_k}})_{k=1\ldots,m}$, where $X_{\underline{e_k}}$
denotes the trait at the end of edge $e_k$.
For example, a weighted average (Fig.~\ref{fig:traitevol}) is
a reasonable approximation for continuous traits
that are influenced by many genes \cite{2018Bastide-pcm-net}.

\begin{figure*}[!t]
    \centering
    \begin{tabular}{cc}
         \subfloat[]{\includegraphics[]{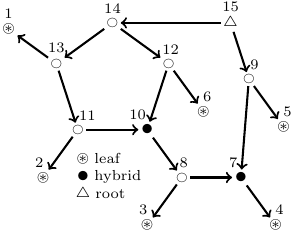} \label{fig:net}
    } &
    \subfloat[]{\includegraphics[]{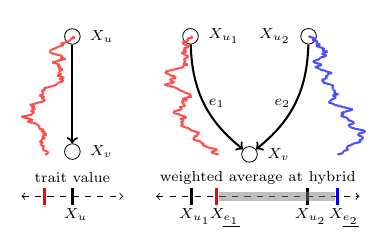} \label{fig:traitevol}
    } \\
    \subfloat[]{\includegraphics[]{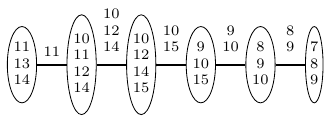} \label{fig:cliquetree}
    } &
    \subfloat[]{\includegraphics[]{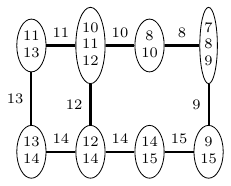} \label{fig:clustergraph}
    } \\
    \end{tabular}
    \caption{(a) Rooted phylogenetic network.
    (b) Example realization of a univariate trait evolution,
    along one edge $uv$ (left) or
    both parent edges of a hybrid node $v$ (right).
    The gray region shows the range of possible trait values at the
    hybrid node under a weighted-average model.
    (c) Clique tree for any graphical model on the network in (a),
    with numbers refering to nodes in (a). Its maximum cluster size is $k^*=4$.
    (d) Loopy cluster graph for the network in (a),
    with maximum cluster size $k=3$ .}
    \label{fig:net-traitevol-clustergraph}
\end{figure*}

\begin{definition}[Weighted-average merging rule {\cite[Eq.~3.3]{teo2024leveraging}}]
    \label{ch4:def:mergingrule}
    For a hybrid node $v$ with parent edges $e_1$, \dots, $e_m$ and inheritance
    weights $\gamma(e_1)$, \dots, $\gamma(e_m)$, the \emph{weighted-average model} assumes that
    \begin{equation}\label{eq:mergingrule-weightedaverage}
        X_v = \sum_{e_k \text{ parent of } v}\gamma(e_k) X_{\underline{e_k}}\;.
    \end{equation}
\end{definition}
If each $e_k$ has a linear Gaussian model \eqref{eq:lineargaussian-1edge},
then the weighted-average model \eqref{eq:mergingrule-weightedaverage} implies a linear
Gaussian model \eqref{eq:lineargaussian-1hybrid} for $\{v\}\cup\pa(v)$.

\subsection{Belief propagation on cluster graphs}\label{ch4:beliefpropagation}
\noindent
Belief propagation is a technique for efficiently computing the marginals of a probability
density expressed as a graphical model. It can be leveraged to compute the likelihood for
parameter optimization.
We focus on directed graphical models.

\begin{definition}[Directed graphical model {\cite[Sec.~4(a)]{teo2024leveraging}}]
    \label{ch4:def:directedgm}
    Let $p_\theta$ be the joint probability density for a set of random vectors
    $\{X_v,v\in V\}$, with parameters $\theta\in\Theta$.
    A \emph{directed graphical model} for $p_\theta$ consists of a DAG $G=(V,E)$ and a set
    $\Phi\coloneq\{\phi_v,v\in V\}$ of non-negative functions such that:
    \begin{enumerate}
        \item $p_\theta=\prod_{\phi_v\in\Phi}\phi_v$
        \item Each $\phi_v$ is proportional to
        the conditional density for $v$'s node family in $G$:
        $\phi_v\propto p_\theta(X_v\mid X_u,u\in\pa(v))$.
    \end{enumerate}
    We say that $p_\theta$ \emph{factorizes} over $G$ and refer to the $\phi_v$ as
    \emph{factors}.
    The set of $X_w$ that $\phi_v$ is defined over is referred to as its \emph{scope},
    denoted $\Scope(\phi_v)$. Finally, $\Scope(\Phi)\coloneq\cup_{\phi_v\in\Phi}\Scope(\phi_v)$.
\end{definition}
Gaussian trait models on a phylogenetic network can typically be cast as a graphical model
with $G=N$ and $\Phi$ consisting of linear Gaussian factors.
Belief propagation then proceeds by passing messages on a cluster graph constructed from the
graphical model.

\begin{definition}[Cluster graph {\cite[Def.~1]{teo2024leveraging}}]
    \label{ch4:def:clustergraph}
    A \emph{cluster graph} for a directed graphical model $(G,\Phi)$ is an undirected graph
    $\mathcal{U}=(\mathcal{V},\mathcal{E})$, whose vertices $\mathcal{C}_i\in\mathcal{V}$,
    $i\in[|\mathcal{V}|]$ are subsets of $\Scope(\Phi)$ and are called \emph{clusters}.
    Additionally, $\mathcal{U}$ must satisfy:
    \begin{enumerate}
        \item Each factor $\phi_v$ can be assigned to a cluster $\mathcal{C}_{\alpha(\phi_v)}$
        that contains its scope, i.e., $\Scope(\phi_v)\subseteq\mathcal{C}_{\alpha(\phi_v)}$;
        $\alpha$ denotes a suitable map from factors to cluster indices.
        \item Each edge $\{\mathcal{C}_i,\mathcal{C}_j\}\in\mathcal{E}$ is labeled with a
        non-empty \emph{separating set} (or \emph{sepset}) $\mathcal{S}_{i,j}\subseteq
        \mathcal{C}_i\cap\mathcal{C}_j$.
        Thus, only clusters with common elements can be joined.
        \item For each $X_v\in\Scope(\Phi)$, $(\mathcal{V}_{X_v}, \mathcal{E}_{X_v})$
        is a subtree of $\mathcal{U}$,
        where $\mathcal{V}_{X_v}\subseteq \mathcal{V}$ and
        $\mathcal{E}_{X_v}\subseteq\mathcal{E}$ are the sets of clusters and edges
        whose labels contain $X_v$.
    \end{enumerate}
    If $\mathcal{U}$ is a tree, then we call $\mathcal{U}$ a \emph{clique tree} and refer to
    its clusters as cliques (Fig.~\ref{fig:cliquetree}).
    For a clique tree, properties 2 and 3 imply that
    $\mathcal{S}_{i,j}=\mathcal{C}_i\cap\mathcal{C}_j$.
    If $\mathcal{U}$ contains one or more cycles, then we call $\mathcal{U}$ a
    \emph{loopy cluster graph} (Fig.~\ref{fig:clustergraph}).
\end{definition}

$\mathcal{U}$ acts as a data structure for ordering the sequence of operations and holding
intermediate computations used to build up various marginals of $p_\theta$.
When an edge is traversed, intermediate computations aggregated at the source cluster are
marginalized, and the result is passed to the target cluster for downstream computations.
The intermediate computations tracked at each cluster and edge are scalar functions that are
defined over their variables and called \emph{beliefs}, denoted $\beta_i(\mathcal{C}_i)$ or
$\mu_{i,j}(\mathcal{S}_{i,j})$.
Cluster beliefs are marginalized to produce \emph{messages}, which update neighbor cluster
beliefs by multiplication, while edge beliefs record the most recent messages passed.
Sending a message across an edge is called a \emph{belief update} (Alg.~\ref{ch4:alg:bp}),
while circulating messages throughout $\mathcal{U}$ is known as
\emph{belief propagation}, or \emph{loopy belief propagation} to emphasize the use of a
loopy cluster graph.
Beliefs are initialized as the product of factors assigned to the corresponding cluster
(Def.~\ref{ch4:def:clustergraph})
and are instantiated at the observed data.
That is, if $X_v=\mathrm{X}_v$ is observed,
then this value is fixed in all beliefs defined over $X_v$, and $X_v$ is removed from their
scopes.
Variables determined by model parameters (e.g., the root state
$X_\rho$ in a Brownian motion trait model) are similarly removed from beliefs' scopes.

\begin{algorithm}
    \caption{Belief update from $\mathcal{C}_i$ to $\mathcal{C}_j$
    \cite[Alg.~1]{teo2024leveraging}}
    \label{ch4:alg:bp}
    \begin{algorithmic}[1]
        \vspace{-.4em}
        \State
        \begin{minipage}[t]{0.45\linewidth}
            Compute message ($\widetilde{\mu}_{i\rightarrow j}$) from source belief
            ($\beta_i$):
        \end{minipage}
        \hfill
        \begin{minipage}[c]{0.5\linewidth}
            \begin{equation*}
                \widetilde{\mu}_{i\rightarrow j}\leftarrow\int\beta_i
                d(\mathcal{C}_i\setminus\mathcal{S}_{i,j})
            \end{equation*}
        \end{minipage}
        \State
        \begin{minipage}[t]{0.45\linewidth}
            Update target ($\beta_j$) and edge ($\mu_{i,j}$) beliefs:
        \end{minipage}
        \hspace{.3em}
        \begin{minipage}[c]{0.5\linewidth}
            \begin{equation*}
                \begin{split}
                    \beta_j^\text{new} &\leftarrow (\widetilde{\mu}_{i\rightarrow j}/\mu_{i,j})\beta_j \\
                    \mu_{i,j}^\text{new} &\leftarrow \widetilde{\mu}_{i\rightarrow j}
                \end{split}
            \end{equation*}
        \end{minipage}
    \end{algorithmic}
\end{algorithm}

If $\Phi$ consists of linear Gaussian factors, then beliefs can be compactly parametrized
and belief updates can be expressed as updates to these parameters
\cite[Sec.~4(c)]{teo2024leveraging}.
This is an instance of \emph{Gaussian belief propagation}, where the graphical model (not
necessarily directed) represents a Gaussian distribution.

\subsection{Calibration and the factored energy}\label{ch4:factoredenergy}
\noindent
$\mathcal{U}$ is said to be \emph{calibrated} if all its beliefs are
\emph{marginally consistent}, or equivalently, if for any pair of neighbors $\mathcal{C}_i$
and $\mathcal{C}_j$:
$\int\beta_i d(\mathcal{C}_i\setminus\mathcal{S}_{i,j})\propto\mu_{i,j}\propto
\int\beta_jd(\mathcal{C}_j\setminus\mathcal{S}_{i,j})$.
If so, belief updates do not change the 
beliefs of a calibrated cluster graph.
For a clique tree, calibration is guaranteed in two traversals: a postorder traversal towards
any designated root clique, followed by a preorder traversal outwards from the root clique.
For a loopy cluster graph, calibration requires multiple traversals, and may not be
guaranteed.

Since neighbor clusters must exchange messages for their beliefs to attain marginal
consistency, a sound \emph{message schedule} (i.e., a sequence of messages to be passed)
should traverse each edge, in both directions.
Since multiple exchanges between neighbor clusters may be needed, a \emph{valid} schedule
\cite{wainwright2003tree} should additionally repeat each message infinitely often.
For example, cycling through a sound finite schedule satisfies this requirement.

Beliefs can be viewed as unnormalized estimates of the conditional distribution of cluster
or edge variables given the observed data $\mathrm{Y}$:
$\beta_i\propto\widehat{p}_\theta(\mathcal{C}_i\mid\mathrm{Y})$ and
$\mu_{i,j}\propto\widehat{p}_\theta(\mathcal{S}_{i,j}\mid\mathrm{Y})$.
Beliefs are exact ($\beta_i\propto p_\theta(\mathcal{C}_i\mid\mathrm{Y})$
and $\mu_{i,j}\propto p_\theta(\mathcal{S}_{i,j}\mid\mathrm{Y})$)
on calibrated clique trees,
but approximate on calibrated loopy cluster graphs.
Upon calibration, the factored energy can be computed to approximate the log-likelihood.

\begin{definition}[Factored energy {\cite[Eq.~5.1]{teo2024leveraging}}]
The \emph{factored energy} of a cluster graph $\mathcal{U}=(\mathcal{V},\mathcal{E})$
with calibrated beliefs $q^*=\{\beta^*_i,\mu^*_{i,j}\mid\mathcal{C}_i\in\mathcal{V},
\{\mathcal{C}_i,\mathcal{C}_j\}\in\mathcal{E}\}$ is defined as
\begin{equation}\label{ch4:eq:fe}
    \begin{gathered}
        \widetilde{F}_\mathcal{U}(q^*) \coloneq \\
        \sum_{\mathcal{C}_i\in\mathcal{V}}\E_{\beta_i^*}(\log\psi_i|_\mathrm{Y})
        + \sum_{\mathcal{C}_i\in\mathcal{V}}\ent(\beta_i^*) -
        \hspace{-.8em}\sum_{\{\mathcal{C}_i,\mathcal{C}_j\}\in\mathcal{E}}
        \hspace{-.8em}\ent(\mu_{i,j}^*)
    \end{gathered}
\end{equation}
where $\psi_i=\prod_{\phi_v\in\Phi,\alpha(\phi_v)=i}\phi_v$ combines the factors assigned
to cluster $i$, $(\cdot)|_\mathrm{Y}$ denotes evaluation at the observed data,
$\E_{\beta_i^*}(\cdot)$ is the expectation with respect to $\beta_i^*$, and $\ent(\cdot)$
is the entropy.

$\widetilde{F}_\mathcal{U}$ resembles the evidence lower bound (ELBO) to the
log-likelihood $\LL=\log(\int p_\theta |_\mathrm{Y}dX)$, where $X$ denotes the
unobserved variables in $\Scope(\Phi)$ \cite{koller2009probabilistic,blei2017variational}.
If $\mathcal{U}$ is a clique tree, then $\widetilde{F}_{\mathcal{U}}=\LL$
exactly\footnote{For a calibrated clique tree, the likelihood can also be obtained by
fully integrating any of the beliefs, e.g., $\int\beta^*_i d\mathcal{C}_i=
\int\mu^*_{i,j}d\mathcal{S}_{i,j}=\int p_\theta|_\mathrm{Y}dX$.\label{footnote:integrateroot}}.
\end{definition}

In the factored energy \eqref{ch4:eq:fe}, each term
is an integral of lower dimension
and cheaper to compute than the likelihood,
$\int p_\theta|_\mathrm{Y}dX$.
The computational cost of each such integral is parametrized by the size of the associated
cluster $|\mathcal{C}|$ or sepset $|\mathcal{S}|$.
Thus, the cost of belief propagation and of evaluating the factored energy are parametrized
by the maximum cluster size $k=\max_{\mathcal{C}\in\mathcal{V}}|\mathcal{C}|$.
A larger $k$ is associated with higher cost.
For a given graphical model,
$k^*$ denotes the minimum $k$ among all possible clique trees for this model.
While the exact LL and calibrated beliefs can be computed with a clique tree
(with $k\ge k^*$ necessarily),
a loopy cluster graph permits $k<k^*$ with smaller clusters and hence cheaper
messages.

\section{Methods}
\noindent
This section describes simulations, using our Julia package \texttt{PhyloGaussianBeliefProp},
that investigate the use of LBP on cluster graphs to fit a Brownian motion (BM) trait model
on a network.
We first compare the runtime and accuracy of using the factored energy to approximate the
log-likelihood, for different maximum cluster sizes $k$.
We then assess the accuracy of parameter estimation by maximizing the factored energy, using
$k$ chosen from the previous step.

\subsection{Maximum factored energy estimation}\label{ch4:methods:mfe}
\noindent
The error for using the factored energy (FE) of a calibrated cluster graph to approximate the
log-likelihood (LL) intuitively improves with larger clusters and sepsets, and a more
treelike topology, but is otherwise not well understood as a function of cluster graph
structure.
Whether or not maximum factored energy (MFE) estimation is a reasonable proxy for maximum
likelihood (ML) estimation depends on the size of the error across parameter space and more
importantly, the extent to which the shapes of the FE and LL surfaces are similar.
Further, whether or not MFE estimation is practical depends on the extent to which cheaper
messages reduce the computational effort needed to obtain FE values.
Though individual messages can be cheaper on a loopy cluster graph than on a clique tree,
calibration requires more traversals.

We use an illustrative case to better understand if and how MFE estimation can be
applied.
We compared MFE and ML estimation for fitting a BM model of $p$-dimensional trait evolution
with variance rate $\bm{\Sigma}\in\mathbb{R}^{p\times p}$ on a phylogenetic network $N=(V,E)$
with $n$ tips and fixed root state $X_\rho=\mu\in\mathbb{R}^p$.
The parameters to be estimated are $\bm{\Sigma}$ and $\mu$.
We focus on this model because it has closed-form ML estimates \eqref{ch4:eq:mle}.

Let $Y_i\in\mathbb{R}^p$ be the unobserved trait vector at leaf $i$, $\bm{Y}=\begin{bmatrix}
Y_1\cdots Y_n\end{bmatrix}\in\mathbb{R}^{p\times n}$, and $Y=\mtov(\bm{Y})\in\mathbb{R}^{np}$.
$Y\sim\mathcal{N}(1_n\otimes\mu,\mathbf{P}_y\otimes\bm{\Sigma})$, where $\mathbf{P}_y\in
\mathbb{R}^{n\times n}$ is the shared-path matrix for the leaves
\cite[Eq.~1]{2018Bastide-pcm-net}, and $1_n$ is the $n$-vector of 1s.
Alternatively, $\bm{Y}\sim\mathcal{MN}_{p,n}(\mu 1_n^\top,\bm{\Sigma},\mathbf{P}_y)$
in the matrix-normal formulation.
The closed-form expressions for the ML estimates and LL of this model (given
observed data $\bm{Y}=\mathbf{Y}$) \cite{revellharmon2008,glanz2018expectation} obviate the
use of BP and loopy BP, to the extent that $\mathbf{P}_y$ can be accurately inverted:
\begin{equation}\label{ch4:eq:mle}
    \begin{aligned}
    \widehat{\mu}_\text{ML} &= \mathbf{Y}\mathbf{P}^{-1}_y1_n/\lVert 1_n
        \rVert_{\mathbf{P}^{-1}_y}^2 \\
    \widehat{\bm{\Sigma}}_\text{ML} &= \lVert\mathbf{Y}^\top-1_n \widehat{\mu}_\text{ML}^\top
        \rVert^2_{\mathbf{P}^{-1}_y}/n \\
    \loglik(\widehat{\mu}_\text{ML},\widehat{\bm{\Sigma}}_\text{ML}) &= -\frac{np}{2}(1+\log2\pi)
        -\frac{p}{2}\log\lvert\mathbf{P}_y\rvert \\
        &\hspace{1.1em} -\frac{n}{2}\log\lvert\widehat{\bm{\Sigma}}_\text{ML}\rvert
    \end{aligned}
\end{equation}
where $\lVert\cdot\rVert_{\mathbf{P}^{-1}_y}^2$ denotes the function $(\cdot)^\top
\mathbf{P}^{-1}_y(\cdot)$ when applied to an input matrix of conformable dimensions.
These equations allow us to conveniently assess the accuracy of the MFE and the
corresponding parameter estimates $\widehat{\mu}_\text{MFE}$,
$\widehat{\bm{\Sigma}}_\text{MFE}$.

\subsection{Spanning trees schedule and regularizing beliefs}\label{ch4:sched&reg}
\noindent
For LBP, we based our message schedule on a small, though not necessarily minimal, set of
spanning trees of the cluster graph, that together cover all its edges and hence all possible
messages \cite[Sec.~5(a)]{teo2024leveraging}.
This set is chosen by iteratively applying Kruskal's algorithm for finding a minimum-weight
spanning tree \cite{cormen2009introduction}, each time incrementing the weights of edges
that are selected for the current spanning tree, and stopping when all edges of the cluster
graph have been used.
The schedule cycles through the spanning trees, and at each step traverses the current
spanning tree in postorder, then preorder.
We refer to each pass through this set of spanning trees as an \emph{iteration} of LBP.
This is a specific implementation of the tree-based reparametrization framework proposed in
\cite{wainwright2003tree}, except that each ``tree update'' does not necessarily calibrate
the current spanning tree (e.g., because some messages are defined over a smaller scope than
the intersection of the source and target clusters, i.e., $\mathcal{S}_{i,j}\subset
\mathcal{C}_i\cap\mathcal{C}_j$, or because some messages are infinite).

A known issue in LBP is that infinite messages
may arise regardless of the chosen schedule.
In step 1 of Algorithm~\ref{ch4:alg:bp}, message $\widetilde{\mu}_{i\rightarrow j}$ is
infinite
if integrating out the required variables from the source belief
evaluates to infinity, thus preventing update and calibration
of the target belief \cite[Sec.~7(a)]{teo2024leveraging}.
This issue is specific to beliefs over continuous variables, which can be
unnormalizable\footnote{The normalizability of a cluster's belief is sufficient but not
necessary for messages computed from that cluster to be finite.},
and is more prominent for directed graphical models since the factors are conditional
densities. Cluster beliefs initialized by multiplying conditional densities
and without any instantiation at the observed data
may be unnormalizable.
Infinite messages can sometimes be avoided by \emph{regularizing} the initial beliefs
while ensuring that the cluster graph still represents the same probability density
$p_\theta|_\mathrm{Y}$ \cite[Sec.~SM-E]{teo2024leveraging}.
For LBP, we applied the regularization strategy from \cite[Alg.~R4]{teo2024leveraging} so
that all cluster beliefs were normalizable before applying our spanning trees schedule.

\subsection{Simulations}

\subsubsection{Network preprocessing}

We used three networks of varying topological complexity, in terms of network level $\ell$
\cite{gambette2009structure} and moralized treewidth $k^*-1$
(Sec.~\ref{ch4:factoredenergy})\footnote{$k^*$ was not
computed exactly, but approximated by the maximum cluster size of a clique tree constructed
using the min-fill heuristic \cite[Sec.~9.4.3]{koller2009probabilistic}.}.
We refer to these as the
Sikora \cite{sikora2019population} ($n=13$, $\ell=6$, $k^*=5$),
Lipson \cite{lipson2020ancient} ($n=12$, $\ell=12$, $k^*=7$) and
M{\"u}ller \cite{muller2022bayesian} ($n=40$, $\ell=358$, $k^*=54$) networks.
For the Sikora and Lipson networks, we used the admixture weights as inheritance weights,
and the drift weights as edge lengths.
Any degree-2 nodes were suppressed, and any length-0 edges were reassigned the minimum
positive edge length.
We avoided length-0 edges because they introduce deterministic factors that require more
sophisticated belief updates \cite[Sec.~SM-F]{teo2024leveraging}.
Such cases will be handled in a future version of \texttt{PhyloGaussianBeliefProp}.
For the M{\"u}ller network, we used the edge lengths provided (which represent calendar time)
and estimated the inheritance weights for each hybrid edge from the recombination
breakpoints provided (see
\href{https://github.com/bstkj/graphicalmodels_for_phylogenetics_code/blob/main/real_networks/muller2022_nexus2newick.jl}{
\texttt{muller2022\_nexus2newick.jl}} in
\url{https://github.com/bstkj/graphicalmodels_for_phylogenetics_code}).

\subsubsection{Simulating trait data}

For each network, we used \texttt{PhyloNetworks} v0.16.4 \cite{2017phylonetworks} to simulate
100 datasets under a univariate $(p=1)$ and multivariate $(p=4)$ Brownian motion trait model
with root mean $\mu_0=0$, assuming a weighted-average model at the hybrid nodes
(Def.~\ref{ch4:def:mergingrule}).
The following variance rates were used: $\sigma^2_0=1$ for univariate data, and for
multivariate data
\[
\bm{\Sigma}_0 = 
    \begin{bmatrix}
        0.8 & -0.71 & -0.8 & 0.49 \\
        \cdot & 0.8 & 0.81 & -0.41 \\
        \cdot & \cdot & 1.1 & -0.4 \\
        \cdot & \cdot & \cdot & 0.5
    \end{bmatrix}
    \;.
\]
The choice of $\bm{\Sigma}_0$ was motivated by real data from \cite{2023Thorson-phylosem}
on four leaf traits: photosynthetic rate, leaf lifespan, leaf mass per area, and leaf
nitrogen content (Appendix~\ref{ch4:sm:pathdiagram}).
The covariance matrix implied by the path diagram in
\cite[Fig.~2 (top)]{2023Thorson-phylosem}, estimated by phylogenetic structural equation
models, is $\bm{\Sigma}_0\cdot 10^{-3}$.
The corresponding correlation matrix between the 4 traits is approximately
\[\hspace{1.3em}\bm{\rho}_0=\begin{bmatrix}
  1 & -0.887 & -0.853 & 0.775 \\
  \cdot & 1 & 0.863 & -0.648 \\
  \cdot & \cdot & 1 & -0.539 \\
  \cdot & \cdot & \cdot & 1
  \end{bmatrix} \;.
\]

\subsubsection{Accuracy and runtime versus maximum cluster size $k$}
\label{ch4:pickk}

For each network and dataset, we ran LBP on cluster graphs generated using join-graph
structuring for different values of $k$, ranging from 3 to $k^*$.
The smallest value was $k=3$ because all networks here are
\emph{bicombining}, where each hybrid node has two parents.
For the Sikora and Lipson networks, we varied $k$
from 3 to 5 and 3 to 7 by increments of 1. 
For the M{\"u}ller network, we varied $k$ from 3 to 20 by increments of 1,
from 25 to 50 by increments of 5,
then from 51 to 54 by increments of 1.
For each task (e.g., LBP for a given network, dataset, and $k$) we evaluated the FE at the
true parameter values $\mu_0$, $\bm{\Sigma}_0$ after calibration was detected or after 50
iterations had passed, whichever came first, then recorded the relative deviation
$\lvert\frac{\FE-\LL}{\LL}\rvert$ and the combined runtime for passing messages and
computing the FE.
Each task was run on a separate physical core, using a single thread, of an Intel Xeon
E5-2680 v3 processor (30M Cache, 2.50GHz).

The computed FE values only approximate the true FE values since calibration was detected
within a specified numerical tolerance, or may not even have been detected;
but \eqref{ch4:eq:fe}
can be well-defined before calibration.
We excluded the runtimes for building the cluster graph, initializing and regularizing
beliefs, and generating a schedule, though these were negligible in comparison.
Since our software is written in Julia \cite{bezanson2017:Julia}, a Just-In-Time compiled
language, compilation latency affects the initial run, but not subsequent runs on similar
inputs. Thus, each task was repeated to record the second runtime.

We plotted relative deviation and runtime against $k$ to choose a $k$ that strikes a good
balance between accuracy and runtime
(Figs.~\ref{ch4fig:pickk_sikoralipson}--\ref{ch4fig:pickk_muller}).
This choice was made separately for each combination of network and trait dimension.

\subsubsection{Parameter estimation with a chosen $k$}
\label{ch4:methods:pickk}

Using $k$ chosen in the previous step, we compared the accuracy of MFE and ML estimation for
each network and trait dimension.
The ML estimates were computed exactly using \eqref{ch4:eq:mle}, while the MFE estimates were
obtained by numerically optimizing an approximate FE (see above) using the L\nobreakdash-BFGS
algorithm (inverse Hessian approximated from past 10 steps) \cite{nocedal2006:numopt} with
Hager-Zhang line search \cite{hager2006:linesearch} and with the FE gradient approximated by
finite differences.
We started the optimization at
$\widehat{\mu}_\text{init} = \frac{1}{n}\sum_{i=1}^n\mathrm{Y}_i$ and
$\widehat{\bm{\Sigma}}_\text{init} = \frac{1}{n\mathrm{h}}\sum_{i=1}^n(\mathrm{Y}_i-
\widehat{\mu}_\text{init})(\mathrm{Y}_i-\widehat{\mu}_\text{init})^\top$, where
$\mathrm{h}=\med_{1\le i\le n}(\mathbf{P}_y)_{ii}$ is a measure of the network's height and
can be obtained in linear time in the number of nodes
\cite[Prop.~2]{2018Bastide-pcm-net}.
These are the ML estimates under a BM model for an ultrametric star tree, for which the leaf
states have equal variances but no phylogenetic correlation.
As above, the FE was computed after calibration was detected or after 50 iterations had
passed.
The maximum number of optimization steps was set to 50, with early termination triggered if,
after any step, the FE increased by less than 0.01\%, or if all coordinates of the
approximated gradient had magnitude less than $10^{-8}$.

We compared the root-mean-square (RMS) of the difference between
the MFE and ML estimates of $\mu_0$,
$\sigma^2_0$ (or $\bm{\Sigma}_0$).
Each task (i.e., FE optimization for a given network and dataset) was run on a
separate physical core, using a single thread, of an Intel Xeon E5-2687 v2 processor (30M
Cache, 2.70GHz).

\section{Results}

\subsection{Simulations}

\subsubsection{Accuracy and runtime versus maximum cluster size $k$}

For the Sikora and Lipson networks, which are moderately complex ($k^*=5,7$
respectively), accuracy and runtime generally improved as $k$ increased
(Fig.~\ref{ch4fig:pickk_sikoralipson}).
The relative deviation $\lvert\frac{\FE-\LL}{\LL}\rvert$ dropped sharply from
slightly under $10^{-3}$
for $k<k^*$ to under $10^{-12}$ at $k=k^*$.
Thus, we picked $k=k^*$ (using a clique tree) to minimize error and runtime.

\begin{figure*}[!h]
  \centering
  \includegraphics[scale=.23]{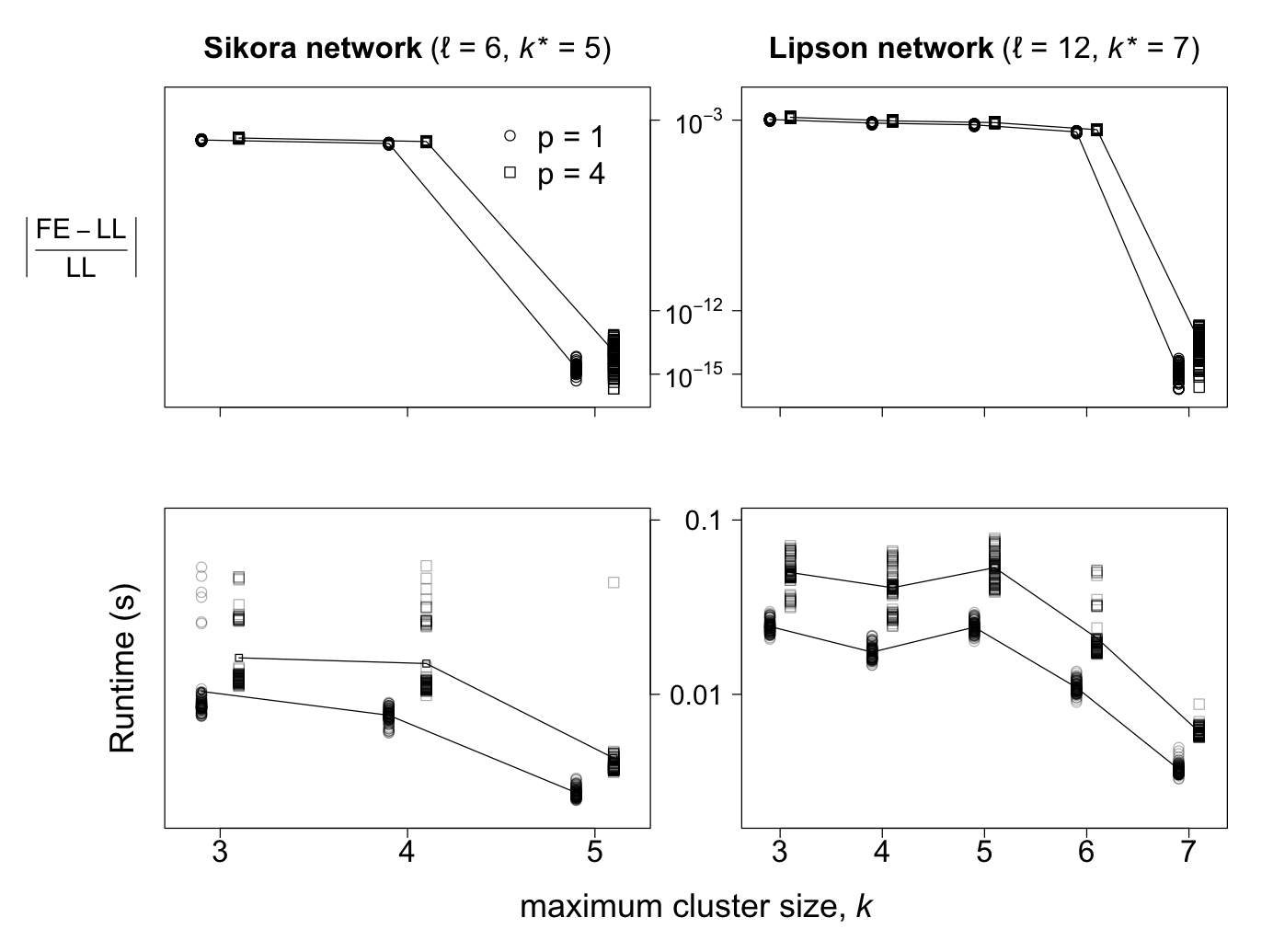}
  \vspace{-1em}
  \caption[Relative deviation and runtime vs maximum cluster size for the Sikora and
    Lipson networks]{Relative deviation and runtime (on the log-scale, vertical axis)
    versus maximum cluster size $k$ for the Sikora and Lipson networks, in the
    univariate and multivariate cases (across 100 datasets each).
    Trajectories of the mean values as $k$ increases are traced.
    Top: relative deviation $\lvert({\FE-\LL})/{\LL}\rvert$.
    Bottom: post-compilation runtimes for passing messages and evaluating the
    FE (upon detection of calibration or after the first 50 iterations).
    For the Lipson network in the multivariate case, the mean number of iterations
    before calibration was detected ($\pm$ 1 standard deviation)
    is 7.87 $\pm$ 0.32 iterations when
    $k=4$ versus 11.22 $\pm$ 0.44 iterations when $k=5$.
    This explains the rise in runtime from $k=4$ to $k=5$.
  }\label{ch4fig:pickk_sikoralipson}
\end{figure*}

For the M{\"u}ller network, which is highly complex ($k^*=54$), accuracy and
runtime generally improved with $k$ in the univariate case, but varied less
trivially with $k$ in the multivariate case (Fig.~\ref{ch4fig:pickk_muller}).
Relative deviation varied similarly in the univariate ($p=1$) and multivariate ($p=4$) cases,
greatly exceeding 1 (100\%) for $k<10$, dropping below 0.1 (10\%) for
$k\ge 11$, and below 0.01 (1\%) when $k\ge 25$ for $p=1$, or when $k\ge 20$ for $p=4$.
It 
decreased rapidly from $k=3$ to $10$, and then more gradually from
$k=11$ to $53$, with a sharp drop (on the log-scale) at $k=k^*=54$.
In contrast, runtime differed markedly between $p=1$ and $p=4$.
For $p=1$, runtime increased slightly from $k=3$ to $35$ then decreased
after $k\ge 35$, therefore we picked $k=k^*$ as the best value.
For $p=4$, runtime increased substantially between $k=16$ and $18$,
so we picked $k=11<k^*$ for $p=4$ to achieve both a shorter runtime and good accuracy.

\begin{figure*}[!t]
  \centering
  \includegraphics[scale=.23]{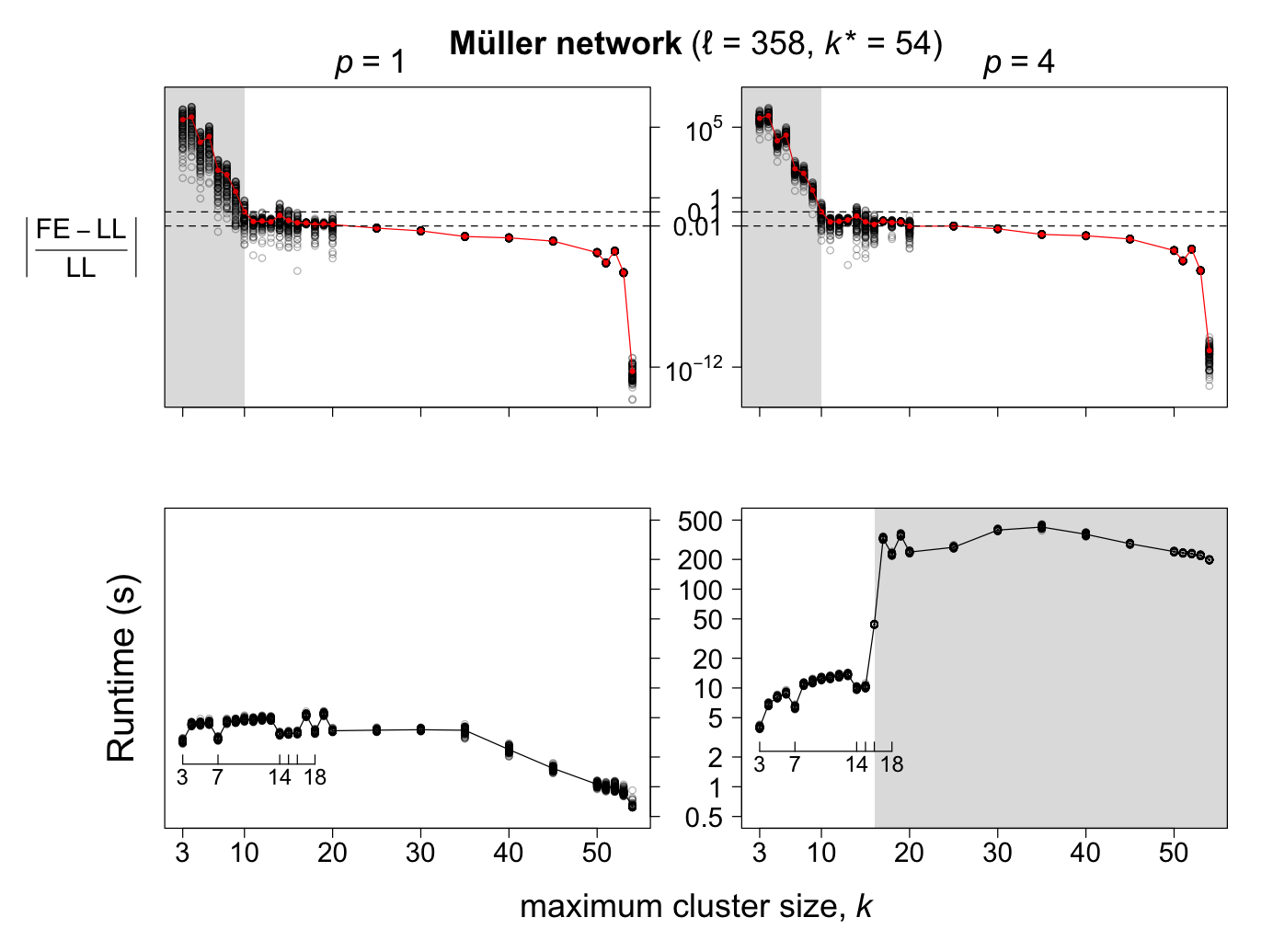}
  \vspace{-1em}
  \caption[Relative deviation and runtime vs maximum cluster size for the M{\"u}ller
    network]{Relative deviation and runtime (on the log-scale, vertical axis)
    versus maximum cluster size $k$
    for the M{\"u}ller network, across 100 datasets.
    Trajectories of the mean values are traced.
    Top: relative deviation $\lvert({\FE-\LL})/{\LL}\rvert$.
    The horizontal lines mark relative deviations of 10\% and 1\%.
    The shaded regions indicate values of $k$ with highly inaccurate FE
    (average error $\gtrsim$ 10\%).
    Bottom: post-compilation runtimes for passing messages and evaluating the
    FE (upon detection of calibration or after the first 50 iterations).
    The shaded regions indicate values of $k$ where runtime is notably worse than
    for other values of $k$.
    In the univariate case (left), calibration was consistently
    detected within 50 iterations after $k=35$.
    In the multivariate case (right), there is a large increase in runtime around $k=16$.
    For $k=3,7,14,15,16,18$ (marked on inset axis),
    the message schedule used two trees to span the cluster graph,
    while for other values of $k$, the schedule used three trees. This explains the
    observed dips in runtime at $k=3,7,14,15,16,18$.
    }\label{ch4fig:pickk_muller}
\end{figure*}

Calibration was consistently detected within 50 iterations across all datasets for the
Sikora and Lipson networks for all $k$, and for the M{\"u}ller network
only for $k\ge 35$, with the number of iterations required generally decreasing in $k$
(Table~\ref{ch4:tab:pickk_calibration}).
We expect runtime to decrease with $k$ in these cases, as the decrease in the
number of passed messages may outweigh the increased cost per message.
When calibration was not detected within 50 iterations, then the runtime was for 50 iterations
exactly, and we expect it to increase with $k$ due to more expensive messages.
The observed runtimes generally conformed with our expectations, with two exceptions:
runtime was not monotone decreasing from $k=3$ to $5$ for the Lipson network,
and not monotone increasing from $k=3$ to $19$ for the M{\"u}ller network.
The former can be explained by differences in the mean number of iterations before
calibration was detected
(7.9 $\pm$ 0.3 iterations when $k=4$ versus 11.2 $\pm$ 0.4 iterations when $k=5$),
while the latter can be explained by differences in the length
of the message schedule (two 
versus three spanning trees).

\begin{table}[h]
    \begin{center}
    \caption{Mean number of iterations ($\pm$ 1 standard deviation) to calibration for the
    M{\"u}ller network at $\mu_0$ and $\sigma^2_0$ ($\bm{\Sigma}_0$).}
    \label{ch4:tab:pickk_calibration}
    \begin{tabular}{|c|c|c|}
        \hline
        & \multicolumn{2}{c|}{mean number of iterations to calibration} \\
        \cline{2-3}
        $k$ & univariate $(p=1)$ & multivariate $(p=4)$ \\
        \hline
        35 & 43.18 $\pm$ 3.40 & 46.38 $\pm$ 1.65 \\
        40 & 24.69 $\pm$ 1.65 & 26.99 $\pm$ 1.40 \\
        45 & 13.34 $\pm$ 0.84 & 13.99 $\pm$ 0.17 \\
        50 & 7 $\pm$ 0.67 & 7.66 $\pm$ 0.48 \\
        51 & 5.85 $\pm$ 0.36 & 6.02 $\pm$ 0.14 \\
        52 & 5.62 $\pm$ 0.51 & 5.98 $\pm$ 0.14 \\
        53 & 4.18 $\pm$ 0.39 & 4.62 $\pm$ 0.49 \\
        54 & 2 & 2 \\
        \hline
    \end{tabular}
\end{center}
\end{table}

\subsubsection{Parameter estimation with a chosen $k$}
\label{ch4:results:params}
Since $k<k^*$ (using a loopy cluster graph) was chosen only for the M{\"u}ller network in the
multivariate case, we only assessed the accuracy of LBP for parameter estimation in that
setting.
We noticed a bimodal distribution in the relative deviation
$\lvert({\widehat{\text{MFE}}-\widehat{\text{ML}}})/{\widehat{\text{ML}}}\rvert$
between the optimized FE ($\widehat{\text{MFE}}$) and the
true maximum LL at the true ML estimate ($\widehat{\text{ML}}$).
The relative deviation was below 5\% 
for a ``good'' group of 67 datasets,
and exceeded 27\% 
for the remaining ``bad'' 33 datasets (Fig.~\ref{ch4fig:opt_muller_mv}).

For the mean parameter, the MFE and ML estimates were similarly far from $\mu_0$,
though close to each other, across both groups (Fig.~\ref{ch4fig:opt_muller_mv2}).
For the variance and correlation parameters, the MFE and ML estimates were similarly close
to $\bm{\Sigma}_0$, $\bm{\rho}_0$ in the ``good'' group, though the MFE estimates could be much worse than
the corresponding ML estimates in the ``bad'' group.
This ``bad'' group led to an increased root-mean-square error for the MFE estimator
compared to the exact ML estimator (Fig.~\ref{ch4fig:opt_muller_mv_rmse}).

\section{Discussion}
\noindent
We showed empirically that the factored energy 
can closely approximate the log-likelihood 
using cluster graphs of maximum cluster size $k$ much smaller than needed for exact
inference using a clique tree ($k^*$).
Further, we demonstrated that such $k$, assessed for good FE approximation to the LL
at the true parameter values only, 
can be effective more generally for parameter estimation by
numerically optimizing the FE.
The resulting MFE estimates were very close
to the theoretical ML estimates a majority of the time.
Theoretical analysis reveals that this good performance is
unsurprising, since the LL surface and the FE surface assuming perfect calibration
are parallel over the parameter space under our simulation model:

\begin{theorem}\label{ch4:thm:surface}
    Consider the Brownian motion model of evolution for a $p$-dimensional trait,
    with root state
    $\mu$ and homogeneous variance rate $\bm{\Sigma}$, on a bicombining phylogenetic network
    with observed data $\mathrm{Y}$ at its leaves.
    Assume that all tree edges have positive length, and each hybrid node
    has at least one parent edge with positive length.
    Let $\mathcal{U}$ be a valid cluster graph for this model,
    as in Def.~\ref{ch4:def:clustergraph}.
    If $\mathcal{U}$ can be calibrated, then its factored energy differs from the
    log-likelihood by an additive constant across all $\mu\in\mathbb{R}^p$ and
    $\bm{\Sigma}\in\mathbb{R}^{p\times p}$ positive-definite.
\end{theorem}
\begin{proof}
    See Appendix~\ref{ch4:sm:mfe}.
\end{proof}

The prospect of reliable MFE estimation is perhaps only relevant for networks of high
complexity, i.e., large $k^*$, so that there is a potential trade-off between
approximation error and runtime.
For simpler networks and smaller trait dimensions, belief propagation on a
clique tree appears to run faster than LBP.

\subsection{Phase transitions for accuracy and runtime with maximum cluster size}
\noindent
Given that the FE can approximate the LL well using $k$ much smaller than $k^*$, it is
natural to ask
(1) if the approximation error
decreases with $k$, and
(2) if there is a clear threshold for $k$ beyond which the error
becomes acceptably low.
In our simulations, the approximation error
appeared to decrease with $k$,
(Figs.~\ref{ch4fig:pickk_sikoralipson}--\ref{ch4fig:pickk_muller}).
Notably, on the highly complex M{\"u}ller network, the FE
approximation error showed two distinct regimes: a fast rate of
decrease followed by a slow rate of decrease,
with a regime transition at $k=11$.
This transition aligned with the threshold $k\ge 10$ to reach a reasonable
approximation with at most 10\% error. For $k<10$, the error exceeded 100\%.
The abruptness of this transition suggests that $k\ge 10$ might correspond to the
the required cluster size for calibration to be possible,
though calibration does not intrinsically guarantee low error.
Indeed, for the Sikora and Lipson networks,
for which the approximation error was always under 0.1\%,
calibration was detected fairly quickly (within 8 iterations).
For the M{\"u}ller network, calibration was detected consistently (within 50 iterations)
only for $k\ge 35$.
Increasing the maximum number of iterations beyond 50
showed that calibration could be consistently detected for $k\ge 10$.
However, more iterations led the FE to diverge further for $k<10$.
This boost in convergence and divergence from more iterations provides
an explanation for the apparent phase transition in the FE accuracy
after the calibration threshold $k=10$.

Runtime generally decreased with $k$ for the Sikora and Lipson networks.
In contrast, runtime generally increased until $k\le 35$ and decreased
after $k\ge 35$ for the M{\"u}ller network.
These trends can mostly be explained in terms of the relative importance, for different
problem sizes (defined by network size and trait dimension), of having to compute more
expensive messages versus having to pass
fewer messages for calibration as $k$ increases.
However, there are aspects of runtime behavior that require deeper explanation.
Notably, for the M{\"u}ller network for $p=4$,
there was a large increase in runtime around $k=16$,
after which runtime was stable.
We suspect that a separate phase transition for runtime is associated with memory latency,
only observed when trait dimension and network complexity,
both of which parametrize the space complexity of (L)BP, are higher.

Assuming that reasonable accuracy cannot be attained before
some calibration threshold $k$, it makes sense to pick
a cluster graph with this $k$ or higher when attempting to find a
good trade-off between accuracy and runtime.
Additionally, knowledge of the problem size and memory resources available should suggest a
practical upper bound for $k$. 

\subsection{Choosing the maximum cluster size $k$}
\noindent
In practice, we need a way to choose $k$ without expending too much computational effort.
For accuracy and runtime considerations, this choice should intuitively depend on network
complexity and trait dimension.
Our simulations suggest that on networks with complexity as high as $k^*=54$ and for
small trait dimensions ($p\le 4$), good accuracy may be attained for $k\approx 10$.
Our results also suggest that using a clique tree can
still be optimal with both lowest runtime and best accuracy
for a univariate trait ($p=1$).
Since the space and time complexities of representing a belief and sending messages are
parametrized by powers of
the dimension $pk$ of the largest cluster's belief precision
(Sec.~\ref{ch4:reliability}), a blunt rule could be to choose
$k=k^*$ if $pk^* \lesssim 60$ and $p\le 2$,
and $k=\min(cp,k^*)$
otherwise, where $c\approx 3$.
This rule interpolates our simulation results.
It favors exact BP when runtime is expected to be feasible, and
otherwise chooses LBP with large-enough clusters to attempt reliable calibration.
Naturally, its applicability is machine-dependent and should be further evaluated.
However, it considers the effects of problem size on runtime, and has the advantage of being
computationally cheap to apply:
$k^*$ is estimated once per network (e.g., using the min-fill heuristic, which has polynomial
complexity in the number of nodes), independently of the observed data.

Choosing a good $k$ is necessary but not sufficient
for accurate parameter estimation.
Another choice to be made is the maximum number of iterations $n_\text{iter}$
of LBP, at each fixed parameter $\theta$,
before $\text{FE}(\theta)$ is computed.
Ideally, the FE should be close to its limiting value after $n_\text{iter}$ iterations.
Fortunately, this may occur prior to calibration.
In examples from \cite[Fig.~6]{teo2024leveraging}, the FE converged faster and 
more smoothly towards its limit than did the cluster graph beliefs.
In our simulations, we used $n_\text{iter}=50$ with reasonable effectiveness,
regardless of $k$.
Alternatively, we could have sought to jointly optimize $k$ and $n_\text{iter}$ for good
runtime and accuracy.
For each $k$ and for some fixed reasonable $\theta$, $n_\text{iter}(k)$ can be chosen
by tracking the convergence of the FE (a computationally expensive task) across iterations of LBP.
The overall performance of different $(k,n_\text{iter}(k))$ can then be compared over
multiple datasets.

\subsection{Making Gaussian LBP useful for phylogenetics}
\noindent
Gaussian LBP opens the door to many applications in phylogenetics,
for analyses of large datasets under complex and heterogeneous evolutionary models
along complex phylogenetic networks, for which no methods currently exist.
The utility of Gaussian LBP for phylogenetics will require
reliable parameter estimation and scalability.

\subsubsection{Reliability}
\label{ch4:reliability}

MFE estimation appeared to be a reliable proxy for ML estimation, for $k$ and $n_\text{iter}$
large enough, in a majority of our simulations.
However, there were practical challenges in computing the FE.
In Gaussian LBP, beliefs and messages are parametrized as log-quadratic forms with a
precision
matrix $\mathbf{K}$ (positive semidefinite)
and a potential
vector $h$ \cite[Sec.~4(c)]{teo2024leveraging}:
\begin{equation}\label{ch4:eq:gbp}
    \begin{aligned}
    \beta_i \Big(\begin{bmatrix}x_\mathrm{S} \\ x_\mathrm{I}\end{bmatrix}\Big)
        &\propto \exp\biggl(-\frac{1}{2}
            \begin{bmatrix}x_\mathrm{S} \\ x_\mathrm{I}\end{bmatrix}^\top
            \overbrace{\begin{bmatrix}\mathbf{K}_\mathrm{S} &
            \mathbf{K}_{\mathrm{S},\mathrm{I}} \\ \mathbf{K}_{\mathrm{S},\mathrm{I}}^\top &
            \mathbf{K}_\mathrm{I}\end{bmatrix}}^{\mathbf{K}}
            \begin{bmatrix}x_\mathrm{S} \\ x_\mathrm{I}\end{bmatrix} \\ 
            &\hspace{3.5em} + \
            \Bigl(\overbrace{\begin{bmatrix}h_\mathrm{S} \\ h_\mathrm{I}\end{bmatrix}}^h\Bigr)^\top
            \begin{bmatrix}x_\mathrm{S} \\ x_\mathrm{I}\end{bmatrix}\biggr) \\
        \widetilde{\mu}_{i\rightarrow j}(x_\mathrm{S}) 
        &\propto \exp\biggl(-\frac{1}{2}
            \lVert x_\mathrm{S}\rVert^2_{\mathbf{K}/\mathbf{K}_\mathrm{I}} \\
            &\hspace{3.5em} + \ (h_\mathrm{S}-
            \mathbf{K}_{\mathrm{S},\mathrm{I}}\mathbf{K}_\mathrm{I}^{-1}h_\mathrm{I})^\top
            x_\mathrm{S}\biggr)
    \end{aligned}
\end{equation}
where $x_\mathrm{S}$ are the variables in $\mathcal{S}_{i,j}$
(the \textbf{S}cope of the message), $x_\mathrm{I}$ are the remaining variables
in $\mathcal{C}_i\setminus\mathcal{S}_{i,j}$ (to be \textbf{I}ntegrated out) and
$\mathbf{K}/\mathbf{K}_\mathrm{I}\coloneq\mathbf{K}_\mathrm{S}-\mathbf{K}_{\mathrm{S},
\mathrm{I}}\mathbf{K}_\mathrm{I}^{-1}\mathbf{K}_{\mathrm{S,\mathrm{I}}}^\top$.
All beliefs must have a positive-definite precision for the FE to be well-defined, and
presumably be close enough to calibration for the LL to be well approximated.
Thus, the FE is typically computed upon multiple traversals of the cluster graph, which makes
the resulting beliefs susceptible to numerical error (e.g., truncation) and hence limits the
numerical precision of the FE.
To send a message, a principal submatrix of the belief precision is inverted
($\mathbf{K}_\mathrm{I}$ above) but numerical error can render this submatrix
near-singular
and prevent the message from being computed.
This hinders calibration since beliefs are unable to ``communicate''
and keeps the FE ill-defined.
This is a crucial issue to deal with when many FE evaluations are required during
numerical optimization, such as to approximate the gradient or to adapt the step size using
line search.

In practice, we attributed failure
to attain positive-definite belief precisions during LBP,
despite prior regularization,
to numerical errors driven by a poor candidate parameter $\theta$ far from the optimal value.
For example, belief precisions are scaled by the variance rate $\bm{\Sigma}$ in the
homogeneous BM model (Appendix~\ref{ch4:sm:fede})
so that large values of $\bm{\Sigma}$ can effectively overwhelm
adjustments made by regularization, which may be much smaller.
To deal with this practical challenge, our implementation returned a very large value for
the FE during optimization (e.g., \texttt{Inf} or $10^{10}$) whenever
an infinite message is encountered or the FE was ill-defined.
This had the effect of tuning down the magnitude of the step size of
the optimization routine to propose the next candidate value $\theta$,
and was effective at steering candidate values towards the right order
of magnitude in our simulations.

In general, it is difficult to anticipate how implementation and input (e.g., regularization
method, trait model, candidate parameter values, network edge lengths)
jointly contribute to numerical errors.
Efforts to improve numerical robustness are further complicated by the
difficulty of checking if a message fails for numerical or theoretical reasons.
Theoretical work is needed to tease these cases apart, e.g., by proving or disproving
that regularization can theoretically guarantee well-defined messages throughout LBP.

\subsubsection{Scalability}
\label{ch4:scalability}

We suggested that LBP may scale better with network complexity than BP, or the naive
approach of computing the inverse covariance for the leaf states to obtain the LL.
Here, we compare the cost of these approaches theoretically
as the data 
and cluster graphs get larger and more complex. 

The cost of computing the LL or FE using (L)BP mainly consists of (i) the
cumulative cost for messages and (ii) the cost of either fully integrating the
root belief to obtain the LL for BP, or evaluating the FE at the final beliefs for LBP
(Appendix~\ref{ch4:sm:complexity}).
The cost of a message is $\mathcal{O}(d^3)$, where $d$ is the dimension of its precision
matrix and is at most $kp$.
Hence, cost (i)
is $\mathcal{O}(s(kp)^3)$, where $s$ is the scheduled number of messages to be passed.
Cost (ii) is $\mathcal{O}((kp)^3)$ for BP and $\mathcal{O}((v_\text{cg}+e_\text{cg})(kp)^3)$
for LBP, where $v_\text{cg}$ and $e_\text{cg}$ are the numbers of clusters and edges in the
cluster graph.
Since the cluster graph is connected 
and at least $2e_\text{cg}$ messages need to be passed
(one in each edge direction) before computing the FE, we have
$v_\text{cg}+e_\text{cg}\le 2e_\text{cg}\le s$.
Thus, the combined cost of (i) and (ii) is $\mathcal{O}(s(kp)^3)$ for (L)BP.
Now using $s^*$ and $k^*$ for BP on a clique tree and keeping $s$ and $k$ for LBP,
the above suggests that $sk^3\ll s^* {k^*}^3$ for LBP to be competitive with BP.

The cost of the naive approach depends on how the
inverse covariance $\mathbf{V}_\text{leaf}^{-1}$ for the leaves is computed.
Let $n$ and $\tilde{n}$
be the number of leaves and internal nodes in the network.
For linear Gaussian networks, $\mathbf{V}_\text{leaf}^{-1}$ can be obtained in
$\mathcal{O}((np)^3)$ as follows.
In a preorder traversal of the network, we can compute the
marginal (co)variances for node $v$:
$\var(X_v)=\bm{V}_v + \bm{q}_v\var(X_{\pa(v)})\bm{q}_v^\top$,
and $\cov(X_v,X_u)=\bm{q}_{v}\cov(X_{\pa(v)},X_u)$ for $u$ listed before $v$,
using notations from Def.~\ref{ch4:def:lineargaussian}.
Assuming that each hybrid node has a bounded number 
of parents (typically $2$),
these operations construct the covariance $\mathbf{V}_\text{all}$ over all
$n+\tilde{n}$ nodes in $\mathcal{O}((n+\tilde{n})^2 p^3)$.
For the BM model with the weighted-average rule, $\bm{q}_v$ has form
$r_v\otimes\mathbf{I}$ for some row vector $r_v$, so the cost
to construct $\mathbf{V}_\text{all}$ drops to
$\mathcal{O}((n+\tilde{n})^2 p^2)$.
Finally, we get get $\mathbf{V}_\text{leaf}$ as a submatrix
of $\mathbf{V}_\text{all}$, and invert it in $\mathcal{O}((np)^3)$.
The costs presented for (L)BP and the naive approach
are driven by matrix inversion, and have similar multiplicative constants
(Appendix~\ref{ch4:sm:complexity}).
For (L)BP to be competitive with the naive approach, the above suggests that the
scheduled number of messages should be $s<(n/k)^3$.

A clique tree constructed from a chordal graph over the same node set as the original
phylogenetic network 
has at most $n+\tilde{n}$ clusters \cite[Prop.~4.16]{golumbic2004:agt},
and join-graph structuring constructs clique trees with
exactly $n+\tilde{n}$ clusters \cite{mateescu2010join}.
In either case, the number of messages $s^*$ is at most $2(n+\tilde{n})$.
Combining with the earlier criterion, $(n+\tilde{n}) \ll (n/k^*)^3$ is favorable for BP
to be competitive with the naive approach.
In a rooted binary network $\tilde{n}=n-1+2h$
\cite[Lem.~2.1]{mcdiarmid2015:counting},
so this criterion simplifies to $(n+h)\ll(n/k^*)^3$.

We illustrate these bounds using the M\"uller network,
whose number of leaves $n=40$ is moderately small,
but with $h=361$ hybrids. 
Our clique tree had $k^*=54$, $e^*_\text{cg}=800$ edges,
and $s^*=1600$ to pass a message in each edge direction.
Our cluster graph with $k=11$ had $e_\text{cg}=1160$
and required $s\approx (980\times 2)\times 3 \times 50$ 
from traversing 3 spanning trees repeatedly $n_\text{iter}=50$ times.
As $sk^3 > s^* {k^*}^3$,
LBP is expected to be slower than BP,
which aligns with our runtime results in the univariate case
(Fig.~\ref{ch4fig:pickk_muller}, bottom-left).
However, both are expected to be slower than the naive
approach for the BM model,
because they require to pass many more messages than
$(n/k^*)^3\approx 0.4$ (for BP) 
or $(n/k)^3\approx 48$ (for LBP). 

The optimal number of scheduled messages $s^*$ is
known for BP on given a clique tree:
$s^*=e_\text{cg}^*$ (from one traversal).
However, for LBP it is generally unclear how to devise schedules that
minimize $s$, while getting close to calibration.
The reliance on suboptimal message schedules limits the utility of LBP.
Adaptive schedules have been proposed to improve the per-message efficiency of attaining
calibration
\cite{elidan2006residualBP,sutton2007residualBP,knoll2015weightdecayBP,aksenov2020parallelRBP},
yet these use more memory and introduce a computational overhead for selecting messages.
Although convergence of the FE is more relevant than calibration for deciding when to stop
passing messages, it is less used since repeated evaluations of the FE are computationally
impractical.

\subsubsection{Possible use cases}
The theoretical bounds above, which do not account for
implementation details such as cost of memory access,
shed some light on possible use cases for LBP to fit
Gaussian trait models on phylogenetic networks.
For example, it may not be easy for (L)BP to decrease runtime compared to the naive approach
since the covariance matrix for the leaves can be efficiently constructed.
This applies to standard models such as the BM
or Ornstein-Uhlenbeck (OU) process, possibly combined with edge length
transformations \cite[Ch.~6]{harmon2019} 
or heterogeneity in the evolutionary parameters
(e.g., the variance rate and optimum for the OU model)
across different sets of edges in the network.
For LBP to be beneficial, both $(n/k)^3$ and $s^*(k^*/k)^3$
have to be large enough to accommodate
multiple traversals of the cluster graph.
These conditions may be less favorable for smaller networks,
which generally have smaller $k^*$ \cite[Fig.~5]{teo2024leveraging}.

In conclusion, our analysis points
to two main settings where LBP could be useful:
large phylogenetic studies on thousands of leaves
\cite{kattge2020:try,tobias2022:avonet,maynard2022:treefunc,peeri2020:mrna_treeoflife}
when their history involves reticulate evolution, and
phylodynamic studies, which increasingly involve large highly-reticulated networks
\cite{suchard2018:beast,muller2020bayesian,muller2022bayesian}.
In particular, the framework it provides is amenable to the use of flexible multi-regime
Gaussian trait models \cite{Mitov2019,brahmantio2025:mixedGaussian} that we expect to see
wider usage of in both settings.

\subsection{Future work}
\noindent
We conclude by listing several theoretical and engineering challenges that could be tackled.
\begin{enumerate}
    \item Do our regularization strategy and spanning trees message schedule, in combination,
    ensure that all messages are theoretically well-defined?
    If not, can alternative belief initializations and schedules provide this
    guarantee, or a stronger one that ensures that belief precisions remain positive-definite
    after every belief update?
    \item How can we best incorporate regularization during message passing to improve
    robustness to numerical error?
    For example, if we know theoretically that a message should be finite,
    but requires inverting a submatrix that is singular due to numerical error,
    we can regularize the sender and sepset beliefs for this
    submatrix to be invertible. This does not approximate the original
    message but allows LBP to proceed,
    and can be applied regardless of whether
    the message is theoretically well-defined.
    Additional heuristics or theory could improve regularization techniques during message
    passing.
    \item We relied exclusively on join-graph structuring
    to construct cluster graphs of a desired maximum cluster size $k$.
    Can modifications or alternatives to join-graph structuring
    be devised to produce cluster graphs with better structural features
    to achieve calibration more often and faster, given a desired $k$?
    For example, larger sepsets tend to be favorable for calibration.
    However, join-graph structuring typically includes multiple size-1 sepsets by construction.
  \item Where to initialize the optimization?
  We used parameter estimates that assume no phylogenetic correlation,
  which adapt to the data yet are fast to calculate.
  A reasonable improvement could be to use the ML estimates assuming the \emph{major tree}
  \textemdash \ the tree obtained by keeping only the parent edge with the largest
  inheritance weight, for each node in the network\footnote{Unlabelled leaves, with no data,
  are pruned from the major tree. These may arise
  if the network is not \emph{tree-child},
  that is, if some internal node has hybrid children only.}
  These estimates can be efficiently computed for the BM model \cite{freckleton2012fast}.
  For other models such as the OU process, even naive estimates that assume no phylogenetic
  correlation require numerical optimization for some parameters.
\end{enumerate}
This non-exhaustive list stresses the equal importance of theory
and implementation for
making this tool more useful to applied researchers.

\section{Acknowledgements}
\noindent
We thank Paul Bastide for helpful feedback on an earlier draft.
This work was supported in part by the University of Wisconsin-Madison
Office of the Vice Chancellor for Research and Graduate Education
with funding from the Wisconsin Alumni Research Foundation;
and by the National Science Foundation through grants
DMS-2023239 to CA and
DMS-1929284 while BT and CA were in residence
at the Institute for Computational and Experimental Research in Mathematics in Providence, RI, during the 
``Theory, Methods, and Applications of Quantitative Phylogenomics" program.

\section*{Software and Data Availability}
\noindent
Data and code for all simulations and analyses are available at
\url{https://github.com/bstkj/adaptingclustergraphs_code}.
In particular,
the Julia package
\href{https://github.com/JuliaPhylo/PhyloGaussianBeliefProp.jl}{\texttt{PhyloGaussianBeliefProp}} was used
at commit
\href{https://github.com/JuliaPhylo/PhyloGaussianBeliefProp.jl/commit/f801d07a4fc4041a40b2df005ff595717ca0936e}{\texttt{f801d07}}.

\bibliographystyle{IEEEtran}
\bibliography{references}

\clearpage
\setcounter{page}{1}

\appendices

\newcounter{appendixfig} 
\renewcommand{\thefigure}{\thesection.\arabic{appendixfig}}
\renewcommand{\thesubsectiondis}{\arabic{subsection}.}
\renewcommand{\thesubsection}{\thesection.\arabic{subsection}}

\section{Covariances from path diagram}
\label{ch4:sm:pathdiagram}
\setcounter{appendixfig}{0}
\noindent
We derive here the trait covariance (rescaled by $10^{-3}$) implied by the path analysis in
\cite{2023Thorson-phylosem}, which we used in our simulation study for the multivariate case
with $p=4$ traits.

\begin{figure}[H]
    \stepcounter{appendixfig}
    \centering
    \includegraphics{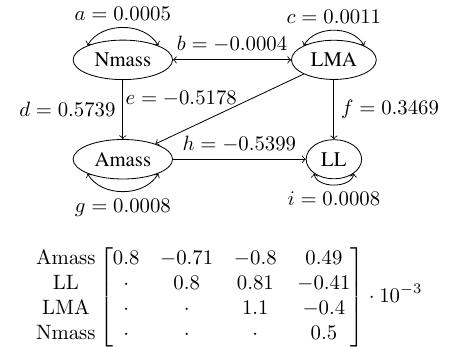}
    \caption[Path diagram of four leaf traits]{Left: Path diagram of four leaf traits: photosynthetic rate (Amass),
    leaf lifespan (LL), leaf mass per area (LMA), leaf nitrogen content (Nmass), reproduced
    from \cite[Fig.~2, top]{2023Thorson-phylosem}.
    Self-loops are labeled by the trait's variance, e.g., $\var(\text{LMA})=0.0011$.
    Bidirectional edges between two traits are labeled by their covariance, e.g.,
    $\cov(\text{Nmass},\text{LMA})=-0.0004$. Undirectional edges from one trait to
    another are labeled by linear dependence weights,
    e.g., $\text{Amass}=d\cdot \text{Nmass} + e\cdot\text{LMA} + \varepsilon$.
    Right: Covariance matrix implied by the path diagram. Each covariance is the sum of
    contributions from all valid paths between the two traits. The contribution of each path
    is determined by its edge weights. Calculations of the covariances from adding up
    distinct path contributions are given below. For example, there are three paths for
    $\cov(\text{Amass},\text{LL})$ and two paths for $\cov(\text{LMA},\text{Amass})$.}
\end{figure}

\begin{align*}
    \cov(\text{Amass},\text{LL}) &= gh + ecf + dbf = -0.00070914 \\
    \cov(\text{LMA},\text{Amass}) &= ce + bd = -0.00079914 \\
    \cov(\text{LMA},\text{LL}) &= cf + ceh + bdh = 0.00081305 \\
    \cov(\text{Nmass},\text{Amass}) &= ad + be = 0.00049407 \\
    \cov(\text{Nmass},\text{LL}) &= adh + beh + bf = -0.00040551
\end{align*}

\section{Likelihood and factored energy surfaces}
\label{ch4:sm:mfe}
\noindent
We consider a bicombining phylogenetic network $N=(V,E)$ with $n$ leaves,
where all tree edges have positive length, and each
hybrid node has at least one parent edge with positive length.
On this network, we assume the BM model for the evolution of a $p$-dimensional trait
with state $\mu$ at the root of the network and
homogeneous variance rate $\bm{\Sigma}$.

To prove Theorem~\ref{ch4:thm:surface},
we derive difference equations for how the FE changes when $\bm{\Sigma}$ is linearly
transformed or when $\mu$ is shifted, and prove their equivalence to the difference
equations for the LL.
This will show that the LL and FE surfaces are parallel,
as claimed in Theorem~\ref{ch4:thm:surface}.

Using notations from section~\ref{ch4:methods:mfe},
$Y=\mtov([Y_1\dots Y_n]) \in\mathbb{R}^{np}$ is the random vector of the leaf states, distributed as
$Y\sim\mathcal{N}(1_n\otimes\mu,\mathbf{P}_y\otimes\bm{\Sigma})$, and $\mathrm{Y}$ is the
corresponding vector of observed data.

\subsection{Log-likelihood difference equations}
\noindent
Let $\LL(\mu,\bm{\Sigma})$ denote the $\LL$ at the parameters $\mu$ and $\bm{\Sigma}$, and
define $\delta_1^{\LL}(\mu_1,\mu_2,\bm{\Sigma})\coloneq
\LL(\mu_1,\bm{\Sigma})-\LL(\mu_2,\bm{\Sigma})$ and
$\delta_2^{\LL}(\mu,\bm{\Sigma}_1,\bm{\Sigma}_2)\coloneq
\LL(\mu,\bm{\Sigma}_1)-\LL(\mu,\bm{\Sigma}_2)$.
The following
equations are easily derived:
\begin{equation}\label{ch4:eq:de_ll_sigma2}
    \delta^{\LL}_2(\mu,\mathbf{A}\bm{\Sigma},\bm{\Sigma})
        = -\frac{1}{2}\lVert\mathrm{Y}-1_n\otimes\mu\rVert_{\mathbf{U}}^2
        -\frac{n}{2}\log |\mathbf{A}|
\end{equation}
\begin{equation}\label{ch4:eq:de_ll_mu}
    \begin{aligned}
        \delta_1^{\LL}(\mu_2,\mu_1,\bm{\Sigma})
        &= -\frac{1}{2}\lVert 1_n\otimes\Delta\mu\rVert^2_{\mathbf{V}_y^{-1}} \\
        +& (1_n\otimes\Delta\mu)^\top\mathbf{V}_y^{-1}(\mathrm{Y}-1_n\otimes\mu_1)
    \end{aligned}
\end{equation}
where $\mathbf{V}_y=\mathbf{P}_y\otimes\bm{\Sigma}$,
$\mathbf{U}=\mathbf{V}_y^{-1}(\mathbf{I}_n\otimes(\mathbf{A}^{-1}-\mathbf{I}_p))$,
$\mathbf{A}\in\mathbb{R}^{p\times p}$ is full-rank such that
$\mathbf{A}\bm{\Sigma}$ is symmetric positive-definite,
and $\Delta\mu = \mu_2-\mu_1$.

\subsection{Factored energy difference equations}\label{ch4:sm:fede}
\noindent
To derive the corresponding quantities for the FE \eqref{ch4:eq:fe},
we consider the energy and entropy terms separately:
\[ 
    \widetilde{F}_\mathcal{U}(q^*) = 
    \overbrace{\sum_{\mathcal{C}_i\in\mathcal{V}}\E_{\beta_i^*}(\log\psi_i|_\mathrm{Y})}^\text{energy term}
    + \overbrace{\sum_{\mathcal{C}_i\in\mathcal{V}}\ent(\beta_i^*)
    - \hspace{-.8em}\sum_{\{\mathcal{C}_i,\mathcal{C}_j\}\in\mathcal{E}}
    \hspace{-.5em}\ent(\mu_{i,j}^*)}^\text{entropy term}
\] 
At calibration, $\sum_{\mathcal{C}_i\in\mathcal{V}}\E_{\beta_i^*}(\log\psi_i)=
\sum_{v\in V\setminus\{\rho\}}\E_{\beta_v^*}(\log\phi_v)$, where $\beta_v^*$ is the belief of
some cluster that contains node $v$ and its parents, $\phi_v$ is the factor for node $v$
(Def.~\ref{ch4:def:directedgm}), and $\rho$ is the root node.
We work with
$\sum_{v\in V\setminus\{\rho\}}\E_{\beta_v^*}(\log\phi_v)$, which is more convenient for our calculations.

\subsubsection{Varying $\mathbf{\Sigma}$ for a fixed $\mu$}
Under the BM model, we have
$X_v\mid X_{\pa(v)}\sim\mathcal{N}((\bm{\gamma}\otimes\mathbf{I}_p)X_{\pa(v)}, \ell_v\bm{\Sigma})$.
If $v$ is a tree node, $\bm{\gamma}=[1]$
and $\ell_v$ is the length of $v$'s parent edge.
If $v$ is a hybrid, $\bm{\gamma}=[\gamma_1 \ \gamma_2]$
contains the inheritance weights of the parent edges of $v$ and
$\ell_v$ is a weighted-average of their lengths,
weighted by their inheritance weights. Then
\begin{equation}\label{ch4:eq:explogfact}
    \begin{aligned}
        \E_{\beta_v^*}(\log\phi_v) &= -\frac{1}{2}\E_{\beta_v^*}(
        \lVert X_{v,\pa(v)}\rVert_{\mathbf{B}_v\otimes(\ell_v\bm{\Sigma})^{-1}}^2
        \\
        &\hspace{1.1em} + \log |2\pi\ell_v\bm{\Sigma}|) \\
        &\hspace{-2em} = -\frac{1}{2}\left(\lVert\widehat{\mu}_{v,\pa(v)}
        \rVert_{\mathbf{B}_v\otimes(\ell_v\bm{\Sigma})^{-1}}^2 + \log |2\pi\ell_v\bm{\Sigma}|\right. \\
        &\hspace{-.9em} + \left.\tr((\mathbf{B}_v\otimes(\ell_v\bm{\Sigma})^{-1})[\mathbf{K}_{\beta_v^*}^{-1}]_{v,\pa(v)})
        \right)
    \end{aligned}
\end{equation}
where $\mathbf{B}_v=[1,-\hspace{-.2em}1]^\top[1,-\hspace{-.2em}1]$ if $v$ is a tree node,
$\mathbf{B}_v=[1,-\hspace{-.2em}\gamma_1,-\hspace{-.2em}\gamma_2]^\top
[1,-\hspace{-.2em}\gamma_1,-\hspace{-.2em}\gamma_2]$ if $v$ is a hybrid node, and
$X_{v,\pa(v)}$ stacks the states of node $v$ on top of those of its parents.
Upon calibration,
$\widehat{\mu}_{v,\pa(v)}=\E(X_{v,\pa(v)}\mid\mathrm{Y})$ does not depend on
$\mathcal{U}$, or on $\bm{\Sigma}$ \cite[Sec.~SM-4]{teo2024:sm}.
$[\mathbf{K}^{-1}_{\beta_v^*}]_{v,\pa(v)}$
denotes the submatrix of $\mathbf{K}_{\beta_v^*}^{-1}$ corresponding to $X_{v,\pa(v)}$, where
$\mathbf{K}_{\beta_v^*}=\mathbf{J}_{\beta_v^*}\otimes\bm{\Sigma}^{-1}$ is the precision of $\beta_v^*$ and $\mathbf{J}_{\beta_v^*}$
does not depend on $\mu$, $\bm{\Sigma}$, or $\mathrm{Y}$ \cite[Lem.~3]{teo2024:sm}. 
Note that
$\tr((\mathbf{B}_v\otimes(\ell_v\bm{\Sigma})^{-1})[\mathbf{K}_{\beta_v^*}^{-1}]_{v,\pa(v)})$
does not depend on $\bm{\Sigma}$.

We are now ready to determine how the energy varies with the variance rate $\mathbf{\Sigma}$.
Letting $\beta^*(\cdot)$ denote $\beta^*$ as a function of $\mathbf{\Sigma}$:
\begin{equation}\label{ch4:eq:energycomp}
    \begin{aligned}
        & \sum_{\mathcal{C}_i\in\mathcal{V}} (\E_{\beta_i^*(\mathbf{A}\bm{\Sigma})}(\log\psi_i) -
            \E_{\beta_i^*(\bm{\Sigma})}(\log\psi_i)) \\
        =& \hspace{-.8em}\sum_{v\in V\setminus\{\rho\}}\hspace{-.5em}(
            \E_{\beta_v^*(\mathbf{A}\bm{\Sigma})}(\log\phi_v) - \E_{\beta_v^*(\bm{\Sigma})}(\log\phi_v)) \\
        =& \hspace{-.8em}\sum_{v\in V\setminus\{\rho\}}\hspace{-.5em}(-\frac{1}{2}\lVert
            \widehat{\mu}_{v,\pa(v)}\rVert_{\mathbf{U}_v}^2
            -\frac{1}{2}\log|\mathbf{A}|)
    \end{aligned}
\end{equation}
where
$\mathbf{U}_v=(\ell_v^{-1}\mathbf{B}_v\otimes\bm{\Sigma}^{-1})(\mathbf{I}_{n_v}
\otimes(\mathbf{A}^{-1}-\mathbf{I}_p))$
does not depend on $\mathcal{U}$, and $n_v$ is the number of rows or columns of $\mathbf{B}_v$
($n_v=3$ if $v$ is hybrid, 2 otherwise).
When there is absorption of evidence in $\phi_v$, e.g., if the value of $X_v$ or $X_{\pa(v)}$
is observed or fixed as a parameter, the effect on \eqref{ch4:eq:explogfact} is that the
trace term changes but remains independent of $\bm{\Sigma}$.
Thus, \eqref{ch4:eq:energycomp} continues to hold in this case.

We now consider the entropy term, starting with the contribution of a single cluster:
\begin{equation}\label{ch4:eq:entropycomp}
    \begin{gathered}
        \ent(\beta_i^*(\bm{\Sigma})) = (1+\log 2\pi)\frac{m_i p}{2} +
            \frac{1}{2}\log |\mathbf{J}_{\beta_i^*}^{-1}\otimes\bm{\Sigma}| \\
        \ent(\beta_i^*(\mathbf{A}\bm{\Sigma})) - \ent(\beta_i^*(\bm{\Sigma})) =
        \frac{m_i}{2}\log |\mathbf{A}|
    \end{gathered}
\end{equation}
where $m_i=|\mathcal{C}_i|$.
\eqref{ch4:eq:entropycomp} holds for edge beliefs $\mu_{i,j}^*$ as well, but with
$m_{i,j}=|\mathcal{S}_{i,j}|$ and $\mathbf{J}_{\mu_{i,j}^*}$ replacing $m_i$ and
$\mathbf{J}_{\beta_i^*}$.
Summing over the contributions from all clusters and edges:
\begin{equation}
    \begin{aligned}
        & \sum_{\mathcal{C}_i\in\mathcal{V}}(\ent(\beta_i^*(\mathbf{A}\bm{\Sigma}))-
        \ent(\beta_i^*(\bm{\Sigma}))) \\
        &- \hspace{-.8em}\sum_{\{\mathcal{C}_i,\mathcal{C}_j\}\in\mathcal{E}}\hspace{-.8em}
            (\ent(\mu_{i,j}^*(\mathbf{A}\bm{\Sigma}))-\ent(\mu_{i,j}^*(\bm{\Sigma}))) \\
        =& \ \frac{1}{2}(\sum_{\mathcal{C}_i\in\mathcal{V}}m_i -\hspace{-.8em}
        \sum_{\{\mathcal{C}_i,\mathcal{C}_j\}\in\mathcal{E}}\hspace{-.8em}m_{i,j})
        \log |\mathbf{A}| \\
        =& \ \frac{|V|-n-1}{2}\log |\mathbf{A}|
    \end{aligned}
\end{equation}
where the last equality is a consequence of the running-intersection property
(Def.~\ref{ch4:def:clustergraph}, condition~3).
Indeed, for each
non-root internal node in $N$, the clusters that contain that node induce a subtree of
$\mathcal{U}$, so that the node appears in one more cluster than sepset.

Let $\widetilde{F}_\mathcal{U}(\mu,\bm{\Sigma})$ denote the $\FE$ at the parameters $\mu$
and $\bm{\Sigma}$, and define $\delta^{\FE}_2(\mu,\bm{\Sigma}_1,\bm{\Sigma}_2) \coloneq
\widetilde{F}_\mathcal{U}(\mu,\bm{\Sigma}_1)-
\widetilde{F}_\mathcal{U}(\mu,\bm{\Sigma}_2)$.
Then combining the energy and entropy terms gives us:
\begin{equation}\label{ch4:eq:de_fe_sigma2}
    \begin{aligned}
        \delta_2^{\FE}(\mu,\mathbf{A}\bm{\Sigma},\bm{\Sigma}) = \hspace{-.7em}\sum_{v\in V\setminus\{\rho\}}
        \hspace{-.5em}-\frac{1}{2}\lVert\widehat{\mu}_{v,\pa(v)}\rVert_{\mathbf{U}_v}^2 -
        \frac{n}{2}\log |\mathbf{A}|
        \;.
    \end{aligned}
\end{equation}

\subsubsection{Varying $\mu$ for a fixed $\mathbf{\Sigma}$}
Define $\delta_1^{\FE}(\mu_1,\mu_2,\bm{\Sigma})\coloneq\widetilde{F}_\mathcal{U}(\mu_1,
\bm{\Sigma})-\widetilde{F}_\mathcal{U}(\mu_2,\bm{\Sigma})$ and let
$\widehat{\mu}_{v,\pa(v)}(\cdot)$ denote $\widehat{\mu}_{v,\pa(v)}$ as a function of
the root state $\mu$.
From \eqref{ch4:eq:entropycomp}, the entropy term of the FE
does not depend on $\mu$. From the energy term in \eqref{ch4:eq:explogfact},
we get:
\begin{multline}\label{ch4:eq:de_fe_mu}
    \delta_1^{\FE}(\mu_2,\mu_1,\bm{\Sigma}) = 
    -\frac{1}{2}\hspace{-.7em}\sum_{v\in V\setminus\{\rho\}}
    \hspace{-.5em}\left(\lVert\widehat{\mu}_{v,\pa(v)}(\mu_2)\rVert_{\mathbf{B}_v\otimes(\ell_v\bm{\Sigma})^{-1}}^2\right.
    \\
    \left.-\lVert\widehat{\mu}_{v,\pa(v)}(\mu_1)\rVert_{\mathbf{B}_v\otimes(\ell_v\bm{\Sigma})^{-1}}^2\right)
    \; .
\end{multline}

\subsection{Comparing the log-likelihood and factored energy surfaces}
\noindent
The difference equations \eqref{ch4:eq:de_ll_sigma2} and \eqref{ch4:eq:de_fe_sigma2} are equal
for a clique tree since the factored energy of a clique tree is equal to the log-likelihood.
Further, the terms in \eqref{ch4:eq:de_fe_sigma2}
do not depend on the cluster graph $\mathcal{U}$ as long
as calibration is attained. Therefore, equality must hold for any calibrated cluster graph,
that is: $\delta_2^{\LL}=\delta_2^{\FE}$.
By the same reasoning, the difference equations \eqref{ch4:eq:de_ll_mu} and
\eqref{ch4:eq:de_fe_mu} must be equal for any calibrated cluster graph,
and $\delta_1^{\LL}=\delta_1^{\FE}$.

Now letting $\delta(\mu,\bm{\Sigma})=\LL(\mu,\bm{\Sigma}) - \widetilde{F}_\mathcal{U}(\mu,\bm{\Sigma})$,
we get
\begin{align*}
    & \delta(\mu+\Delta\mu,\mathbf{A}\bm{\Sigma}) \\
    =& \ (\LL(\mu,\mathbf{A}\bm{\Sigma}) + \delta^{\FE}_1(\mu+\Delta\mu,\mu,\mathbf{A}\bm{\Sigma}))
    - \widetilde{F}_\mathcal{U}(\mu+\Delta\mu,\mathbf{A}\bm{\Sigma}) \\
    =& \ \LL(\mu,\mathbf{A}\bm{\Sigma}) - \widetilde{F}_\mathcal{U}(\mu,\mathbf{A}\bm{\Sigma}) \\
    =& \ (\LL(\mu,\bm{\Sigma}) + \delta^{\FE}_2(\mu,\mathbf{A}\bm{\Sigma},\bm{\Sigma}))
    - \widetilde{F}_\mathcal{U}(\mu,\mathbf{A}\bm{\Sigma})
    =
    \ \delta(\mu,\bm{\Sigma})
\end{align*}
which implies that $\delta$ is a constant function over $\mu$ and $\bm{\Sigma}$ positive definite.
Thus, the FE is equal to the LL up to an additive constant as claimed
in Theorem~\ref{ch4:thm:surface}, and optimizing either quantity is
theoretically equivalent for this homogeneous BM model.
In practice however, the FE is computed from beliefs that are only approximately
calibrated, and this may
translate to different optimization landscapes for both quantities.

\section{Complexity estimates}\label{ch4:sm:complexity}
\noindent
We evaluate here the cost of BP, LBP, and inverting the covariance for the leaves
to calculate the LL or FE, in terms of number of floating-point operations.

\subsection{Cost of passing a message}
\noindent
From \eqref{ch4:eq:gbp}, a message $\widetilde{\mu}_{i\rightarrow j}$ is obtained by
marginalizing out variables $x_\mathrm{I}$ from the scope $x=\mtov([x_\mathrm{S} \ x_\mathrm{I}])$
of a belief $\beta_i$ with parameters
\begin{equation*}
    \begin{split}
        \mathbf{K} &= \begin{bmatrix}\mathbf{K}_\mathrm{S} & \mathbf{K}_\mathrm{S,I} \\
            \mathbf{K}_\mathrm{S,I}^\top & \mathbf{K}_\mathrm{I}
        \end{bmatrix}, \ h=\begin{bmatrix}h_\mathrm{S} \\ h_\mathrm{I}\end{bmatrix} \mbox{ and } g,
    \end{split}
\end{equation*}
where $\exp(g)$ is the constant of proportionality, that is,
$\beta_i(x)=\exp(-\lVert x\rVert_\mathbf{K}^2/2 + h^\top x + g)$.
Computing $\widetilde{\mu}_{i\rightarrow j}$ requires computing
the message parameters \cite[Alg.~2]{teo2024leveraging}:
\begin{equation*} 
    \begin{split}
        \mathbf{K}_\text{msg} &= \mathbf{K}_\mathrm{S} -
        \mathbf{K}_\mathrm{S,I}\mathbf{K}_\mathrm{I}^{-1}\mathbf{K}_\mathrm{S,I} \\
        h_\text{msg} &= h_\mathrm{S} - \mathbf{K}_\mathrm{S,I}\mathbf{K}_\mathrm{I}^{-1}h_\mathrm{I} \\
        g_\text{msg} &= g + (\log|2\pi\mathbf{K}_\mathrm{I}^{-1}| +
        \lVert h_\mathrm{I}\rVert_{\mathbf{K}_\mathrm{I}^{-1}})/2
        \;.
    \end{split}
\end{equation*}
Computing $\mathbf{K}_\mathrm{I}^{-1}$ from $\mathbf{K}_\mathrm{I}\in\mathbb{R}^{d_\mathrm{I}
\times d_\mathrm{I}}$ using an LU decomposition uses $(5/3)d_\mathrm{I}^3+
\mathcal{O}(d_\mathrm{I}^2)$ flops \cite[Sec.~3.1--3.2]{golub2013:matrix}, where $d_\mathrm{I}\le kp$, $k$ is the maximum cluster
size, and $p$ is the trait dimension.
Given $\mathbf{K}_\mathrm{I}^{-1}$, only the matrix product
$\mathbf{K}_\mathrm{S,I}\mathbf{K}_\mathrm{I}^{-1}\mathbf{K}_\mathrm{S,I}^\top$ in the
computation of $\mathbf{K}_\text{msg}$ may have above-quadratic complexity in $kp$.
For $\mathbf{K}_\mathrm{S,I}\in\mathbb{R}^{d_\mathrm{S}\times d_\mathrm{I}}$, the cost of
sequentially evaluating
$((\mathbf{K}_\mathrm{S,I}\mathbf{K}_\mathrm{I}^{-1})\mathbf{K}_\mathrm{S,I}^\top)$ uses
$2(d_\mathrm{S}d_\mathrm{I}^2+d_\mathrm{S}^2 d_\mathrm{I})$ flops \cite[Tab.~1.1.2]{golub2013:matrix}, which equals $(kp)^3/2$
in the worst case when $d_\mathrm{S}=d_\mathrm{I}=kp/2$.
Thus, the cost of computing a message is at most $(5/3+1/2)(kp)^3\approx 2(kp)^3$ if the lower-order terms are
dropped.
This upper bound applies to BP and LBP regardless of the specific cluster and sepset involved,
though the average cost per message depends on the whole distribution of cluster and sepset
sizes and the message schedule, and is impractical to evaluate.

\subsection{Cost of evaluating the FE}
\noindent
From Appendix~\ref{ch4:sm:fede}, the FE is the sum of an energy term and an entropy term.
Cluster beliefs contribute to the energy and entropy, and edge beliefs only
contribute to the entropy.
Their contributions are independent and their costs add up to the total cost of
evaluating the FE.
For a given cluster or edge, let $\exp(-\lVert x\rVert^2_{\mathbf{K}_0}/2 + h^\top_0 x + g_0)$
be the initial belief before regularization,
$\exp(-\lVert x\rVert^2_{\mathbf{K}}/2 + h^\top x + g)$ be the final belief before the FE is
computed, and $d$ be the dimension of $\mathbf{K}$. 

The energy of a cluster belief is $\exp(-\lVert\mathbf{K}^{-1}h\rVert^2_{\mathbf{K}_0}/2 +
\tr(\mathbf{K}_0\mathbf{K}^{-1}) + h_0^\top(\mathbf{K}^{-1}h)+g_0)$.
Computing $\mathbf{K}^{-1}$ from $\mathbf{K}$ using an LU decomposition uses $(5/3)d^3+\mathcal{O}(d^2)$ flops.
Given $\mathbf{K}^{-1}$, only $\tr(\mathbf{K}_0\mathbf{K}^{-1})$ potentially has
above-quadratic complexity in $d$. However, this term can be computed in $\mathcal{O}(d)$
flops by considering only the diagonal entries of $\mathbf{K}_0\mathbf{K}^{-1}$.
Hence, the cost of evaluating the energy term of the FE is at most $(5/3)v_\text{cg}(kp)^3$,
where $v_\text{cg}$ is the number of clusters in the cluster graph, if the lower-order terms
are dropped.

The entropy of a belief is $-\log(|\mathbf{K}/(2\pi e)|)/2$.
Given an LU decomposition for $\mathbf{K}$, $|\mathbf{K}|$ can be computed in $\mathcal{O}(d)$
flops. Otherwise, computing the decomposition to obtain the determinant uses
$(2/3)d^3+\mathcal{O}(d^2)$ flops.
Hence, assuming that the energy term of the FE has been computed and every cluster belief has
been LU decomposed, the additional cost of evaluating the entropy term of the FE is at most
$(2/3)e_\text{cg}(kp)^3$, where $e_\text{cg}$ is the number of edges of the cluster graph,
if the lower-order terms are dropped.
Thus, the total cost of evaluating the FE is $\mathcal{O}((v_\text{cg}+e_\text{cg})(kp)^3)$.

\onecolumn
\section{Simulation study: supplemental figures}
\setcounter{appendixfig}{0}
\noindent
We provide here more details on the simulations under the most
complex M{\"u}ller network, for a multivariate trait ($p=4$).
\begin{figure*}[h]
    \stepcounter{appendixfig}
    \centering
    \includegraphics[scale=.25]{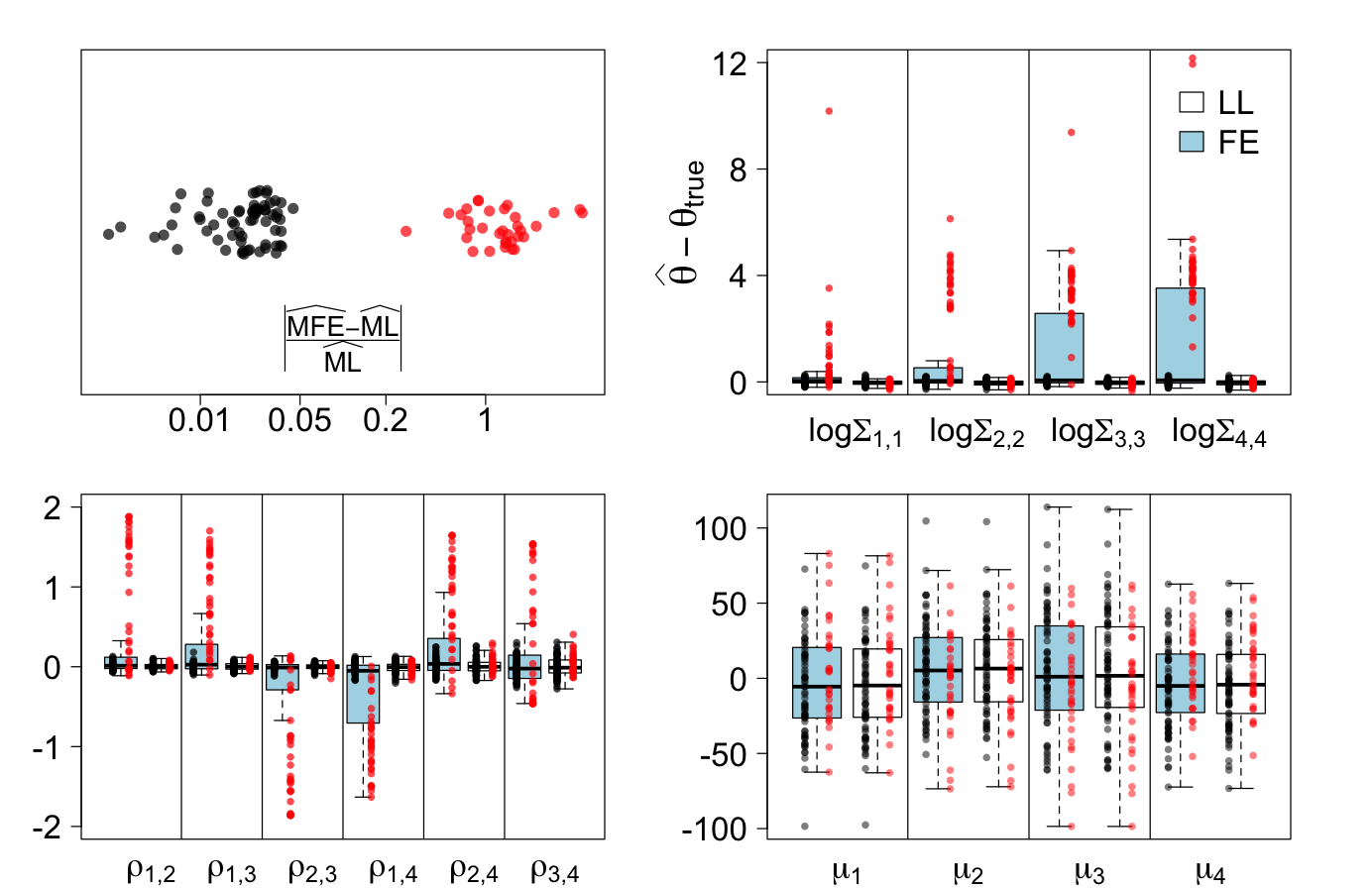}
    \caption[Residuals for MFE and ML estimates for the M{\"u}ller network]{
    Top left: Relative deviation (on the log-scale, horizontal axis) between
    the maximum FE attained ($\widehat{\text{MFE}}$) and
    the theoretical maximum LL ($\widehat{\text{ML}}$)
    for the M{\"u}ller network in the multivariate case,
    using a cluster graph with $k=11$.
    The 100 replicates can be clearly divided into 67 ``good'' datasets (black)
    for which the relative deviation $|({\widehat{\text{MFE}}-\widehat{\text{ML}}})/{\widehat{\text{ML}}}|$
    is below $0.05$, and 33 ``bad'' datasets (red) for which it is large, over $0.27$.
    Top right: estimation error from each optimization objective
    (LL vs FE) for the $p=4$ variance parameters (diagonal terms in $\bm{\Sigma}$),
    on the log-scale.
    Bottom:  estimation error
    for the remaining parameters: 6 correlations $\bm{\rho}$ and
    4 ancestral values $\bf{\mu}$.
    Within each boxplot, the 100 datasets are shown by individual points separated
    into the ``good'' (black) and ``bad'' (red) groups defined in the top-left panel.
    For the variance and correlation parameters,
    estimates based on numerically maximizing the FE are noticeably less precise and
    less accurate when the optimized FE is a poor approximation
    of $\widehat{\text{ML}}$ (red points).
    }\label{ch4fig:opt_muller_mv}
\end{figure*}

\clearpage
\begin{figure}
    \stepcounter{appendixfig}
    \centering
    \includegraphics[]{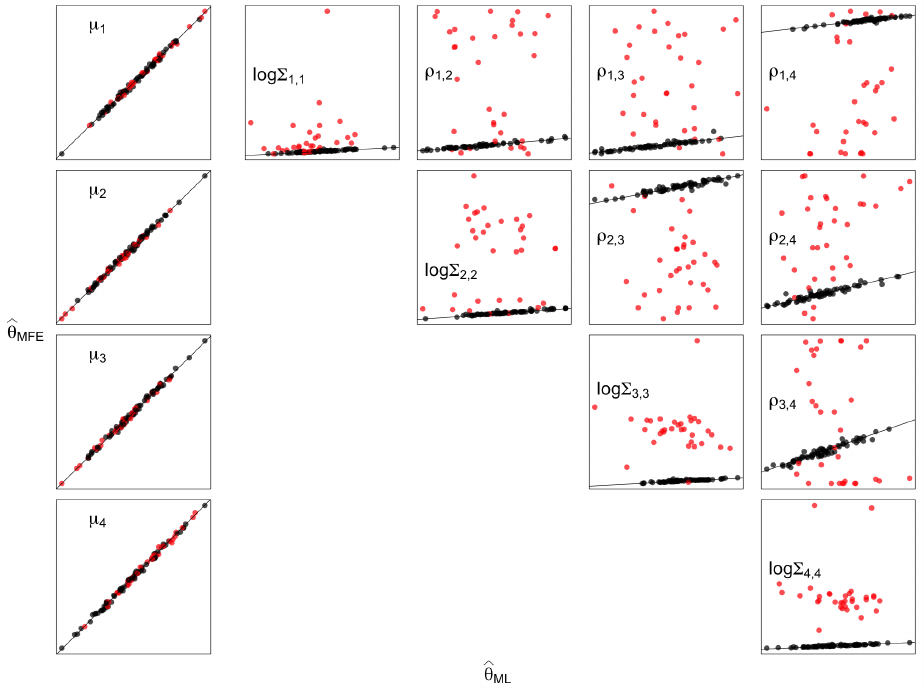}
    \caption[]{Parameter estimates for the M{\"u}ller network
    in the multivariate case using a cluster graph with $k=11$,
    estimated by either numerically maximizing the FE (vertical axis)
    or maximizing the exact LL (horizontal axis).
    Axes are not necessarily on the same scale across different plots
    (axes markings are suppressed for readability) but in each plot,
    the line corresponds to $y=x$ where both estimates are equal.
    Points are colored as in Fig.~\ref{ch4fig:opt_muller_mv}, based on whether
    $|({\widehat{\text{MFE}}-\widehat{\text{ML}}})/{\widehat{\text{ML}}}|$
    is below $0.05$ (black) or above $0.27$ (red).
    The MFE estimates align well with the ML estimates in the former case,
    but not so in the latter case.
    }
    \label{ch4fig:opt_muller_mv2}
\end{figure}

\begin{figure}[!t]
    \stepcounter{appendixfig}
    \centering
    \includegraphics[]{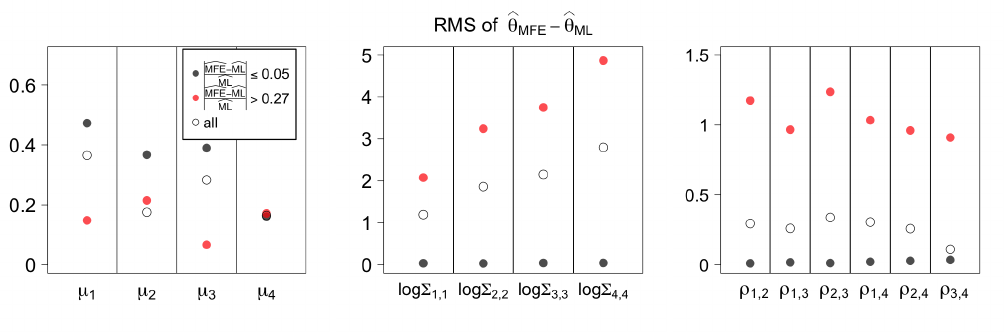}
    \vspace{-1.5em}
    \caption[RMS of the differences between the MFE and the ML estimates for the M{\"u}ller network]{
    Root mean square (RMS) of the differences between the MFE ($\widehat{\theta}_{\text{MFE}}$)
    and ML ($\widehat{\theta}_{\text{ML}}$) estimates for the BM model parameters
    under the M{\"u}ller network in the multivariate case, using a cluster graph with $k=11$.
    $\widehat{\mu}_i$ and $\widehat{\bm{\Sigma}}_{i,i}$ are the estimates of the
    ancestral mean and variance rate for trait $i$.
    $\bm{\rho}_{i,j}$ is the evolutionary correlation between traits $i$ and $j$,
    from $\bm{\Sigma}$, and $\widehat{\bm{\rho}}_{i,j}$ is its estimate.
    For each parameter, the RMS
    of $\widehat{\theta}_{\text{MFE}} - \widehat{\theta}_{\text{ML}}$
    is computed from all 100 replicates, then recomputed for the
    ``good'' subset (67 replicates) and ``bad'' subset (33 replicates),
    for which
    $|({\widehat{\text{MFE}}-\widehat{\text{ML}}})/{\widehat{\text{ML}}}|$
    was below $0.05$ or above $0.27$ respectively, as in Fig.~\ref{ch4fig:opt_muller_mv}.
    Left: for $\mu$, the RMS difference between the two estimators
    is well below their RMSE (RMS between their estimate and the true value),
    which are all above 28. 
    Middle: for the log variances $\log{\bm{\Sigma}}_{i,i}$, the RMSE
    ranges in $[0.09,0.11]$ for the ML estimator, 
    and in $[1.18,2.79]$ for the MFE estimator. 
    The RMS difference between the two estimators is much smaller for the ``good''
    subset than for the ``bad'' subset, showing that the bad subset is driving
    the increase in the RMSE of the FE estimator compared to the exact ML estimator.
    Right: for the correlations ${\rho}_{i,j}$, the bad subset is also driving
    an increase in the RMSE of the FE estimator (ranging in $[0.52,0.73]$) 
    compared to the exact ML estimator (whose RMSEs are in $[0.03,0.13]$). 
    }
    \label{ch4fig:opt_muller_mv_rmse}
\end{figure}

\end{document}